\documentclass[11pt]{article} 
\usepackage[margin=1in]{geometry} 
\pdfoutput=1 

\usepackage{silence} 
\usepackage[dvipsnames,svgnames,table]{xcolor} 

\usepackage{amsmath,amsfonts,amssymb,amsthm}
\usepackage{mathtools}

\usepackage{mmap}
\usepackage[T1]{fontenc}

\usepackage[utf8]{inputenc}
\usepackage{csquotes}
\usepackage[english]{babel}

\usepackage[
    activate={true,nocompatibility},
    final, 
    tracking=true,
    kerning=true,
    spacing=true,
    factor=1100,
    stretch=12,
    shrink=10
    ]{microtype}
\microtypecontext{spacing=nonfrench}
\usepackage[backend=biber,
              style=alphabetic,
        maxbibnames=99,
      minalphanames=3,
      maxalphanames=4,
            backref=true]{biblatex}

\def\anon{0} 
\def\papertitle{Oracle Separations for the Quantum-Classical
Polynomial Hierarchy} 

\if\anon1
\usepackage{hyperref}
\hypersetup{
    pdftitle={\papertitle}, 
    pdfauthor={Anonymous submission}, 
    colorlinks=true, 
    linkcolor=blue, 
    citecolor=blue, 
    urlcolor=green 
}
\else
\usepackage{hyperref}
\hypersetup{
    pdftitle={\papertitle}, 
    pdfauthor={Avantika Agarwal, Shalev Ben{-}David}, 
    colorlinks=true, 
    linkcolor=blue, 
    citecolor=blue, 
    urlcolor=green 
}
\fi

\AtEveryBibitem{
 \clearlist{address}
 \clearfield{date}
 \clearlist{location}
 \clearfield{month}
 \clearfield{day}
 \clearfield{series}
 \clearfield{pages}
 \clearlist{organization}
 \clearfield{number}
 \clearlist{language}

 \ifentrytype{book}{}{ 
  \clearlist{publisher}
  \clearname{editor}
  \clearfield{issn}
  \clearfield{isbn}
  \clearfield{volume}
 }
}

\DeclareFieldFormat{eprint:eccc}{
\mkbibacro{ECCC}\addcolon\space\ifhyperref
    {\href{http://eccc.hpi-web.de/report/#1/}{\nolinkurl{#1}}}
    {\nolinkurl{#1}}}
\DeclareFieldAlias{eprint:ECCC}{eprint:eccc}
\DeclareFieldFormat{eprint:iacr}{Cryptology ePrint\addcolon\space\ifhyperref
    {\href{https://eprint.iacr.org/#1}{\nolinkurl{#1}}}
    {\nolinkurl{#1}}}
\DeclareFieldAlias{eprint:IACR}{eprint:iacr}
\DeclareFieldFormat{eprint:arxiv}{arXiv\addcolon\space\ifhyperref
    {\href{https://arxiv.org/abs/#1}{\nolinkurl{#1}}}
    {\nolinkurl{#1}}}

\DeclareFieldFormat[article]{title}{#1}
\DeclareFieldFormat[inproceedings]{title}{#1}
\DeclareFieldFormat[incollection]{title}{#1}
\DeclareFieldFormat[online]{title}{#1}

\renewbibmacro{in:}{}

\newcommand{\lName}{1}

\newcommand{\donothing}[1]{#1}

\newcommand{\JACM}{\if\lName1\donothing{Journal of the {ACM}}\else{JACM}\fi}
\newcommand{\SICOMP}{\if\lName1\donothing{{SIAM} Journal on Computing}\else{SICOMP}\fi}
\newcommand{\ToC}{\if\lName1\donothing{Theory of Computing}\else{ToC}\fi}
\newcommand{\ToCGS}{\if\lName1\donothing{Theory of Computing Graduate Surveys}\else{ToC}\fi}
\newcommand{\TOCT}{\if\lName1\donothing{{ACM} Transactions on Computation Theory}\else{TOCT}\fi}
\newcommand{\ToIT}{\if\lName1\donothing{{IEEE} Transactions on Information Theory}\else{TOCT}\fi}
\newcommand{\CCjournal}{\if\lName1\donothing{Computational Complexity}\else{CC}\fi}
\newcommand{\CJTCS}{\if\lName1\donothing{Chicago Journal of Theoretical Computer Science}\else{CJTCS}\fi}

\newcommand{\TCS}{\if\lName1\donothing{Theoretical Computer Science}\else{TCS}\fi}
\newcommand{\IPL}{\if\lName1\donothing{Information Processing Letters}\else{IPL}\fi}
\newcommand{\JCSS}{\if\lName1\donothing{Journal of Computer and System Sciences}\else{JCSS}\fi}

\newcommand{\RSA}{\if\lName1\donothing{Random Structures and Algorithms}\else{RSA}\fi}
\newcommand{\JCTA}{\if\lName1\donothing{Journal of Combinatorial Theory, Series A}\else{JCTA}\fi}
\newcommand{\JCTB}{\if\lName1\donothing{Journal of Combinatorial Theory, Series B}\else{JCTB}\fi}
\newcommand{\PJM}{\if\lName1\donothing{Pacific Journal of Mathematics}\else{PJM}\fi}
\newcommand{\QICjournal}{\if\lName1\donothing{Quantum Information and Computation}\else{QIC}\fi}
\newcommand{\IJQI}{\if\lName1\donothing{International Journal of Quantum Information}\else{IJQI}\fi}
\newcommand{\PRA}{\if\lName1\donothing{Physical Review A}\else{PRA}\fi}
\newcommand{\PRL}{\if\lName1\donothing{Physical Review Letters}\else{PRL}\fi}
\newcommand{\VLDB}{\if\lName1\donothing{International Journal on Very Large Data Bases}\else{VLDB}\fi}
\newcommand{\SIDMA}{\if\lName1\donothing{{SIAM} Journal on Discrete Mathematics}\else{SIDMA}\fi}

\DefineBibliographyStrings{english}{%
  backrefpage = {p{.}},
  backrefpages = {pp{.}},
}

\def\bibnonempty{0}
\let\oldcite\cite
\renewcommand{\cite}[1]{\global\def\bibnonempty{1}\oldcite{#1}}

\newtheorem{theorem}{Theorem}

\newtheorem{lemma}[theorem]{Lemma}
\newtheorem{proposition}[theorem]{Proposition}
\newtheorem{corollary}[theorem]{Corollary}
\newtheorem{definition}[theorem]{Definition}

\theoremstyle{definition}

\newcommand{\eq}[1]{\hyperref[eq:#1]{(\ref*{eq:#1})}}
\renewcommand{\sec}[1]{\hyperref[sec:#1]{Section~\ref*{sec:#1}}}
\newcommand{\thm}[1]{\hyperref[thm:#1]{Theorem~\ref*{thm:#1}}}
\newcommand{\lem}[1]{\hyperref[lem:#1]{Lemma~\ref*{lem:#1}}}
\newcommand{\defn}[1]{\hyperref[def:#1]{Definition~\ref*{def:#1}}}
\newcommand{\prop}[1]{\hyperref[prop:#1]{Proposition~\ref*{prop:#1}}}
\newcommand{\cor}[1]{\hyperref[cor:#1]{Corollary~\ref*{cor:#1}}}
\newcommand{\fig}[1]{\hyperref[fig:#1]{Figure~\ref*{fig:#1}}}
\newcommand{\tab}[1]{\hyperref[tab:#1]{Table~\ref*{tab:#1}}}
\newcommand{\alg}[1]{\hyperref[alg:#1]{Algorithm~\ref*{alg:#1}}}
\newcommand{\app}[1]{\hyperref[app:#1]{Appendix~\ref*{app:#1}}}
\newcommand{\conj}[1]{\hyperref[conj:#1]{Conjecture~\ref*{conj:#1}}}
\newcommand{\chap}[1]{\hyperref[chap:#1]{Chapter~\ref*{chap:#1}}}

\DeclareMathOperator{\poly}{poly}
\DeclareMathOperator{\polylog}{polylog}

\newcommand{\B}{\{0,1\}}

\newcommand{\sB}{\{\pm 1\}}

\DeclareMathAlphabet{\mathbbold}{U}{bbold}{m}{n}

\DeclareMathOperator*{\E}{\mathbb{E}} 
\DeclareMathOperator{\bR}{\mathbb{R}}
\DeclareMathOperator{\bN}{\mathbb{N}}

\newcommand{\rhoin}{\rho_{\mathsf{init}}}

\DeclareMathOperator{\R}{R}
\DeclareMathOperator{\Q}{Q}
\DeclareMathOperator{\C}{C}

\DeclareMathOperator{\s}{s}
\DeclareMathOperator{\bs}{bs}
\DeclareMathOperator{\fbs}{fbs}

\newcommand{\adeg}{\widetilde{\deg}}

\renewcommand{\P}{\mathsf{P}}
\newcommand{\NP}{\mathsf{NP}}
\newcommand{\PH}{\mathsf{PH}}
\newcommand{\QCMA}{\mathsf{QCMA}}
\newcommand{\QCMAH}{\mathsf{QCMAH}}
\newcommand{\QCPH}{\mathsf{QCPH}}
\newcommand{\QMA}{\mathsf{QMA}}
\newcommand{\QCS}{\mathsf{QC}\Sigma}
\newcommand{\AC}{\mathsf{AC}}
\newcommand{\BQP}{\mathsf{BQP}}
\newcommand{\BPP}{\mathsf{BPP}}
\newcommand{\MA}{\mathsf{MA}}
\newcommand{\PSPACE}{\mathsf{PSPACE}}
\newcommand{\PP}{\mathsf{PP}}

\newcommand{\QCSigma}{\mathsf{QC\Sigma}}
\newcommand{\QSigma}{\mathsf{Q\Sigma}}
\newcommand{\QPi}{\mathsf{Q\Pi}}
\newcommand{\QSigmapure}{\mathsf{pureQ\Sigma}}

\newcommand\QPH{\mathsf{QPH}}
\newcommand\QEPH{\mathsf{QEPH}}
\newcommand\EXP{\mathsf{EXP}}
\newcommand\QPHpure{\mathsf{pureQPH}}

\newcommand{\mpoly}{\mathrm{mpoly}}
\newcommand{\Bpoly}{\BQP_{/\mpoly}}

\usepackage{xspace}

\newcommand\newmathabbrev[2]{\newcommand{#1}{\ensuremath{#2}\xspace}}
\newcommand\cfont\mathsf

\newmathabbrev\SIP{\cfont{SIPSER}}

\newmathabbrev\QAC{\cfont{QAC}}

\newcommand{\QCPi}{\mathsf{QC\Pi}}

\newcommand{\ayes}{A_{\textup{yes}}} 
\newcommand{\ano}{A_{\textup{no}}} 
\DeclarePairedDelimiter\ket{\lvert}{\rangle}

\newcommand{\set}[1]{{\left\{#1\right\}}}    

\numberwithin{theorem}{section}

\begin{document}


\title{\papertitle}

\if\anon1
\author{Anonymous submission}
\else
\author{
    Avantika Agarwal\\
    \small Institute for Quantum Computing\\
    \small University of Waterloo\\
    \small \texttt{avantika.agarwal@uwaterloo.ca}
    \and
    Shalev Ben{-}David\\
    \small Institute for Quantum Computing\\
    \small University of Waterloo\\
    \small \texttt{shalev.b@uwaterloo.ca}
}
\fi

\date{}
\maketitle

\begin{abstract}
We study the quantum-classical polynomial hierarchy, $\QCPH$,
which is the class of languages solvable by a constant number
of alternating classical quantifiers followed by a quantum
verifier. Our main result is that $\QCPH$ is infinite relative
to a random oracle
(previously, this was not even known relative to any oracle).
We further prove that higher levels of
$\PH$ are not contained in lower levels of $\QCPH$ relative
to a random oracle; this is a strengthening of the
somewhat recent result
that $\PH$ is infinite relative to a random oracle
(Rossman, Servedio, and Tan 2016).

The oracle separation requires lower bounding a certain type
of low-depth alternating circuit with some quantum gates.
To establish this, we give a new switching lemma for quantum
algorithms which may be of independent interest. Our lemma
says that for any $d$, if we apply a random restriction
to a function $f$ with quantum query complexity
$\Q(f)\le n^{1/3}$,
the restricted function becomes exponentially close
(in terms of $d$) to a depth-$d$ decision tree.
Our switching lemma works even in a
``worst-case'' sense, in that only the indices to be restricted
are random; the values they are restricted to are chosen
adversarially. Moreover, the switching lemma also works
for polynomial degree in place of quantum query complexity.
\end{abstract}

{\scriptsize\tableofcontents}
\clearpage

\section{Introduction}

In classical complexity theory, an important complexity class is the
polynomial hierarchy, $\PH$. This is a generalization of $\NP$ to higher
depth: it can be written as the union
$\NP\cup \NP^{\NP}\cup \NP^{\NP^{\NP}}\cup\dots$,
and corresponds to languages that can be
computed using a constant number of alternating quantifiers over
certificates. A problem is in $\PH$ if it can be computed
in polynomial time in the presence of two computationally-unbounded
provers, one of which wants to convince the verifier the input
is a yes-input, and one of which wants to convince the verifier
the input is a no-input, with the provers exchanging
polynomially-sized public messages for a constant number of rounds.
The polynomial hierarchy can therefore be viewed as the class of
problems solvable by an audience member sitting in a debate between
experts, where the debate has the inefficient format 
of alternating between the two speakers only a constant number
of times, even though the total debate time is polynomial in the
input size (polynomially many alternations would result in the
larger class $\mathsf{PSPACE}$).

The polynomial hierarchy can be viewed as a union of different 
``levels'', where the $d$-th level corresponds to debates with
$d$ alternations. The hierarchy is said to \emph{collapse} to
level $d$ if level $d+1$ can solve no more problems than level $d$;
in that case, all higher levels can also be simulated with $d$
rounds of debate. It is widely believed that the polynomial
hierarchy is infinite, meaning it does not collapse to any level.
This is a generalization of the $\mathsf{P}\ne\mathsf{NP}$ conjecture,
which is equivalent to the assertion that $\mathsf{PH}$ does not
collapse to the $0$-th level.

Since proving the polynomial hierarchy is infinite is beyond
current techniques, one may ask instead for oracle separations.
It turns out $\mathsf{PH}$ is infinite even relative to a
\emph{random} choice of oracle (with probability $1$),
though this result is somewhat recent.

\begin{theorem}[\cite{HRST17}]\label{thm:PHInfinite}
$\PH$ is infinite relative to a random oracle.
\end{theorem}

We study a quantum version of the polynomial hierarchy known
as the quantum-classical hierarchy (introduced in \cite{gharibianQuantumGeneralizationsPolynomial2018}). This is the class of problems
solvable by a quantum audience member of a
constant-round debate (with classical messages) between experts.
This class, denoted $\QCPH$, generalizes $\QCMA$ but not $\QMA$
(since the proofs received are classical rather than quantum).
More formally, $\QCPH$ is the union $\bigcup_{i\in\bN} \QCS_i$,
where $\QCS_i$ is defined as the class of languages
for which there is a quantum verifier $V$ satisfying
\begin{align*}
    x\in\ayes \Rightarrow \exists y_1 \forall y_2 \ldots Q_i y_i \colon \Pr[V(x, y_1, y_2, \ldots, y_i)] \geq 2/3 \\
    x\in\ano \Rightarrow \forall y_1 \exists y_2 \ldots \overline{Q_i} y_i \colon \Pr[V(x, y_1, y_2, \ldots, y_i)] \leq 1/3
\end{align*}
for all input strings $x$, where the $y_j$ represent
polynomially-sized classical strings, $Q_i$ is a quantifier,
and $\overline{Q_i}$ is the opposite quantifier.

The class $\QCPH$ is interesting on its own, but another motivation
for its study is the connection to quantum switching lemmas.
Oracle separations for $\PH$ generally reduce to the problem of
giving lower bounds on the classical circuit class $\AC^0$,
consisting of circuits of constant depth. Quantum versions of
constant-depth circuits are of interest because they help model
quantum devices with many qubits but few layers of gates.
Lower bounds on such circuit classes are often shown using
switching lemmas, which assert that certain types of functions
must greatly simplify under a random restriction of their bits;
these switching lemmas can therefore be useful both for the study
of near-term quantum devices and for oracle separations for
complexity classes such as $\QCPH$.

\subsection{Our results}

Our main result generalizes \thm{PHInfinite} to the quantum setting.

\begin{theorem}\label{thm:QCPH}
The following holds relative to a random oracle with probability $1$.
For any constant $d \geq 1$, level $d+1$ of the polynomial hierarchy,
$\Sigma_{d+1}$,
is not contained in level $d$ of the quantum-classical hierarchy,
$\QCPi_d$. In particular, $\QCPH$ is infinite, and no fixed level
of the $\QCPH$ hierarchy contains all of $\PH$.
\end{theorem}

Previously to this work, it was not even known whether
there was \emph{any} oracle relative to which $\QCPH$ is infinite,
let alone a random oracle. Previous switching lemmas for quantum
algorithms, such as the one in
\cite{AIK22}, are insufficient to prove such an oracle separation.\footnote{Note that most quantum circuit lower bounds
are incomparable to ours, since they consider different models of quantum circuits. Our switching lemma is stronger than the comparable
switching lemma of \cite{AIK22}.}
We also note that this theorem implies \thm{PHInfinite} as a
special case, since it also implies that $\mathsf{PH}$ is
infinite relative to a random oracle.

The study of $\PH$ relative to a random oracle boils down to the study
of $\AC^0$ circuits against the uniform distribution.
This, in turn, is usually done using random restrictions and the
switching lemma, together with the random projection technique
introduced in \cite{HRST17}. To prove
\thm{QCPH}, we need a new switching lemma for quantum query
complexity.

\begin{theorem}\label{thm:Qswitching}
Let $N\in\bN$ be sufficiently large, and
let $f$ be a possibly partial Boolean function on $N$ bits
with $\Q(f)\le N^{1/3}$. Let $p=N^{-10/11}$ (we can choose $p \leq O(\frac{1}{Q(f)^2N})^{4/7}$), and consider a
random restriction of $f$ which fixes each bit with probability $1-p$.
Then for any $d\le N^{1/3}$, this restriction is approximated
by a decision tree of height $d$ in the following sense.

For an input $x$ and a choice of unrestricted bits $S$, we define a restriction which sets bits in $\overline{S}$ to according to $x$. Then for every input $x$, there is an ensemble of decision trees
$\{D_{x,S}\}_{x\in\B^n,S\subseteq[N]}$, all of height at most $d$,
with $D_{x,S}$ acting on the unrestricted input bits which are strings of length $|S|$, such that
the following holds: if $S\subseteq[N]$ is chosen at random
with each index $i\in[N]$ included in $S$ independently
with probability $p$, then
\[\forall x,y\in\B^n\;\;
\Pr_S[D_{x,S}(y|_S)\ne f_{x_{\overline{S}}}(y|_S)]< e^{-d^{1/5}}.\]
Here $f_{x_{\overline{S}}}$ denotes the restriction of $f$ to the
partial assignment which fixes the bits of $x$ outside of $S$.%
\footnote{If $f_{x_{\overline{S}}}$ is undefined on input $y|_S$, we count it as equality holding in the probability, meaning that
on inputs outside the domain the decision tree is allowed to
output anything.}
\end{theorem}



This is a type of switching lemma for quantum query complexity,
though its statement can be confusing. We make a few clarifying
comments. First, note that with the parameter $p=N^{-10/11}$,
a function $f$ with quantum query complexity at most $N^{1/3}$ becomes
constant with high probability under a random restriction.
However, the probability of the function becoming constant
is merely $1-1/\poly(N)$, which is not small enough.
The point of the theorem is to approximate the function $f$
to exponentially small error
(in the height of the approximating decision tree).

The random restriction in \thm{Qswitching} is \emph{worst-case} in the
sense that while the positions to be fixed are chosen randomly,
the bits to which those positions are fixed are specified by
a string $x$, which can be chosen arbitrarily. Moreover,
the resulting restricted function $f|_{x_{\overline{S}}}$ must
be approximated by a decision tree even on worst-case choices of
input $y|_S$ (except that the choice of the worst-case $y$
cannot depend on the random choice of $S$). From such a worst-case
statement, we can easily derive a more familiar average-case
switching lemma.

\begin{corollary}\label{cor:UniformSwitching}
For sufficiently large $N$, let $f$ be such that $\Q(f)\le N^{1/3}$,
let $p=N^{-10/11}$, and let $d\le N^{1/3}$. If $f_\rho$ is a random
restriction in which each bit remains unfixed with probability $p$
(and is fixed randomly to $0$ or $1$ otherwise), then
with probability at least $1-e^{-d^{1/10}}$, there is a decision
tree of height at most $d$ which computes $f_\rho$ on a
$1-e^{-d^{1/10}}$ fraction of its inputs.
\end{corollary}

\thm{Qswitching} can be viewed as a version of \cor{UniformSwitching}
which works even for non-uniform distributions. We also strengthen \thm{Qswitching}
so that it works for approximate degree
instead of just quantum query complexity.

\begin{theorem}\label{thm:DegSwitching}
\thm{Qswitching} still holds if the condition on $f$ is
replaced by $\adeg(f)\le N^{1/3}$ instead of $\Q(f)\le N^{1/3}$.
\end{theorem}

This can be viewed as establishing a random-restriction version of the Aaronson-Ambainis
conjecture, which asserts that low degree polynomials can be
approximated by shallow decision trees against the
uniform distribution. Our version works only after a random
restriction is applied to the polynomial, but it works with extremely
strong parameters\footnote{Note that Aaronson-Ambainis conjecture is about bounded real-valued functions rather than Boolean functions; our switching lemma also works in that setting.%
}.
Another (incomparable) version of the
Aaronson-Ambainis conjecture for random restrictions was given
in \cite{Bha24}.

Although we don't explicitly show it,
\thm{DegSwitching} can be strengthened further, so that
it works with a measure known in the literature as
(the square root of)
``critical fractional block sensitivity'', which lower bounds
approximate degree \cite{ABK21}. This measure also
lower bounds the positive-weight quantum adversary bound,
so our switching lemma also works for that measure.

Finally, we give an application of our switching lemmas to
yet another type of oracle separation: we show that the ``$\QCMA$
hierarcy,'' that is,
\[\QCMA\cup \QCMA^{\QCMA}\cup\QCMA^{\QCMA^{\QCMA}}\cup\dots\]
is also infinite relative to a random oracle.
Explicitly, defining $\QCMAH_1=\QCMA$
and $\QCMAH_{d+1}=\QCMA^{\QCMAH_d}$ for $d\ge 1$, we have
the following result.

\begin{theorem}
The following holds relative to a random oracle with probability
$1$. For any constant $d\ge 1$, $\QCMAH_{d+1}$ is not contained
in $\QCMAH_d$. Moreover, no fixed level $\QCMAH_d$ contains
all of $\PH$.
\end{theorem}

Proving lower bounds on $\QCMAH$ relative to an oracle amounts
to proving lower bounds on constant-depth
circuits which have alternating
layers of polynomial-fanin AND gates, and of ``gates''
consisting of quantum query algorithms which make few quantum
queries. This demonstrates that our techniques can be useful
for proving lower bounds on shallow quantum circuit classes.

\subsection{Our techniques}

\paragraph{Random oracle separations.}
To prove our oracle separation for $\QCPH$
relative to a random oracle,
it suffices to show that depth $d$ $\AC^0$ circuits with
``quantum query complexity gates'' at the bottom layer cannot
compute some function in computable in depth $d+2$ classical $\AC^0$,
against the uniform distribution.
The quantum query complexity
gates means that at the bottom of the $d$
alternating layers of AND and OR gates lie Boolean functions
which can each be computed in $\polylog N$ quantum queries.
(For the $\QCMAH$ separation,
we need to allow quantum query complexity
gates in the middle of the circuit as well, not just the bottom.)
This circuit separation implies the random oracle separation
through a standard technique sometimes called ``slow diagonalization''
(see \app{diagonalization}).

To construct a function which is computable in depth $d+2$
classically but not in depth $d$ quantumly (against the uniform
distribution), we use the construction of \cite{HRST17},
which separated classical $\AC^0$ circuits of different depths.
Their proof used random projections, which we also employ.
We must modify their construction to account for the extra layer
of quantum gates, which we will need to ``strip away'' in the
analysis using a random projection and our quantum switching lemma.

In order to establish a depth-hierarchy theorem for circuits, the restrictions used need to ensure that an AND-OR tree retains structure (as in \cite{DBLP:conf/stoc/Hastad86} and \cite{HRST17}) while a smaller-depth circuit simplifies. Thus the sequence of restrictions/projections used needs to alternate between heavily biased towards $1$ and heavily biased towards $0$ (so that not all AND/OR gates are set to constant). In the quantum switching lemma of \cite{AIK22}, the quantum query algorithm becomes close to a DNF with respect to the uniform distribution, when the restriction is also sampled uniformly. We will however need to sample a projection from an underlying distribution different from the one with respect to which we compare the closeness of the resulting quantum query algorithm and DNF (because of the alternating choice of distributions described earlier). This mismatch between the two
distributions requires us to use our ``worst-case'' type of quantum
switching lemma.

Our separation for the $\QCMAH$ hierarchy works similarly,
but requires carefully changing the parameters to ensure the
structure is maintained even in the presence of intermediate
quantum query gates. See \sec{QCMAH} for details.

\paragraph{Quantum switching lemma.}
The proof of our quantum switching lemma relies on an
\emph{adaptive} application of the quantum hybrid method, which
may be of independent interest. Given a quantum algorithm $Q$
acting on a string $x$, a standard hybrid argument
of \cite{BBBV97} says there will only be a small set of positions
in the string $x$ that the algorithm $Q$ ``looks at'';
the output of $Q$ can only be sensitive to a change in a bit of $x$
at one of those few positions. Call those the heavy positions of $x$.
Now, if some of those heavy positions of $x$ are indeed changed,
then not only can the output of $Q$ change, but even the set of
heavy positions of the input can change.

This poses a problem for us: we would like to restrict the function
to the bits of $x$, except at a few random positions. If some
of those random positions are heavy (with respect to $x$), then
the quantum algorithm can still depend on them in a nontrivial
manner. We could try to mimic this by a classical algorithm
which queries those few non-fixed heavy bits, but the problem
is that this is not sufficient to fix the output of $Q$:
the output of the algorithm $Q$ may now depend on new heavy bits.
(It is not clear if a classical decision tree can find these
new heavy bits, since a quantum algorithm may try to query
the rest of the unfixed bits in superposition.)

We get around this problem by applying the hybrid method
iteratively, in an adaptive manner. Beginning with a string $x$,
we find its heavy bits, and randomly choose a few of them to leave
unfixed; we replace those bits with values from a different string
$y$. This gives us a hybrid string $x^{(1)}$ which mostly contains
bits of $x$, but where a few heavy positions were replaced by bits
of $y$. We then iterate this process: we find the heavy bits of
$x^{(1)}$ which are not heavy bits of $x$, sample a few of those
positions at random, and replace them with more bits from $y$,
resulting in $x^{(2)}$. We continue this way and terminate
when there are no new heavy bits.

In each round, there is a constant
(or better) probability of having no new heavy bits which are
unfixed, since the number of heavy bits is small and the probability
of leaving a bit unfixed is also small. This means the number
of rounds cannot be too large except with exponentially small
probability. Once the process terminates, the set of all heavy
bits encountered along the way cannot be too large
(except with exponentially small probability), and can be
used to compute the randomly-restricted function on most inputs.

\paragraph{Polynomial lower bound.}
Since our quantum switching lemma relies fundamentally on the
hybrid method, which talks about
``where the quantum algorithm queried the string'', it may seem
surprising that we can generalize our result to a switching
lemma for polynomials as well. To do this, we rely on a property
shown in \cite{ABK21} (expanding on earlier work by \cite{Tal13})
which says that approximate degree is lower bounded by fractional
block sensitivity. Fractional block sensitivity, in turn, is
equal to fractional certificate complexity, a measure which
assigns to each bit of each input a non-negative weight:
a way of saying ``how much did the polynomial look at position $i$
when given string $x$ as input''.

In fact, we will need a stronger version of the result of
\cite{ABK21}; we adapt their method to prove this. Our result
is the following theorem, which may be of independent
interest.

\begin{theorem}\label{thm:fbsdeg}
Let $f\colon\B^n\to[0,1]$ be a real-valued function which
can be expressed as a polynomial of degree at most $d$.
Then there is an assignment of weights
$\{c_{x,i}\}_{x\in\B^n,i\in[n]}$ to inputs $x$ and positions
$i\in[n]$ such that for all $x,y\in\B^n$, we have
\[\sum_{i:x_i\ne y_i} c_{x,i}\ge |f(x)-f(y)|,\qquad
\sum_{i=1}^n c_{x,i}\le \frac{\pi^2}{4} d^2.\]
\end{theorem}

The proof is similar to the one in \cite{ABK21}; it uses
strong duality of linear programming to convert fractional
certificates to fractional block sensitivity, and uses composition
of $f$ with a version of promise-OR to convert fractional block
sensitivity to regular block sensitivity; the latter can be
turned into sensitivity using a projection, and results in
approximation theory can be used to relate sensitivity to polynomial
degree.

\subsection{Previous work on quantum polynomial hierarchies}
The study of quantum variants of polynomial hierarchy was initiated in \cite{gharibianQuantumGeneralizationsPolynomial2018}. In this section, we survey the results which are already known about these quantum variants of $\PH$.

 We start by comparing $\QCPH$ to an alternative definition of quantum-classical $\PH$, considered in \cite{j.lockhartQuantumStateIsomorphism2017}, which we will call $\QCPH'$. Define $\QCSigma'_1 = \exists \cdot \BQP$ and $\QCSigma'_i = \exists \cdot \QCPi'_{i-1}$. While $\QCSigma'_1$ looks similar to $\QCMA$, it is not known whether $\QCSigma'_1 = \QCMA$ due to the $\exists \cdot \BPP$ versus $\MA$ phenomenon (see \cite{gharibianQuantumGeneralizationsPolynomial2018} for a detailed discussion). This phenomenon refers to the following problem: in $\exists \cdot \BPP$, for every choice of proof $y$, the $\BPP$ verifier must either accept or reject with high probability (for example, it can not accept with probability 1/2). This is not the case with $\MA$, where (in the yes case) the $\BPP$ verifier is only supposed to accept with high probability for at least one choice of proof $y$, but can behave however it likes on other choices of $y$. For this reason, it is not clear whether $\QCPH = \QCPH'$, and in this work (as in \cite{gharibianQuantumGeneralizationsPolynomial2018}) we consider $\QCPH$ to be the natural quantum-classical analogue of $\PH$, because it generalizes $\QCMA$.

The most basic property of $\PH$ is the collapse theorem, which says that if $\P = \NP$, then $\PH = \P$. The collapse property is also known to be true for $\QCPH$:
\begin{proposition}[Collapse theorem for $\QCPH$ \cite{agkr_qph}, \cite{FGN23}]\label{prop:qcphcollapse}
    If $\QCSigma_i = \QCPi_i$, then $\QCPH = \QCSigma_i$.
\end{proposition}
\noindent It is now natural to ask how the collapse property of $\PH$ and $\QCPH$ are related, i.e., whether collapse of one implies the collapse of other. It is easy to see that relative to a $\PSPACE$ oracle, $\P = \NP = \BQP = \QCMA = \PSPACE$. It is known that the collapse property for $\PH$ and $\QCPH$ is not related in a black-box manner:
\begin{proposition}[Corollary 3.7 of \cite{FR99}]
    There exists an oracle $O$ relative to which $\P^O= \BQP^O$ and $\PH^O$ is infinite.
\end{proposition}
\begin{proposition}[Theorem 32 of \cite{AIK22}]
    There exists an oracle $O$ relative to which $\P^O = \NP^O \neq \BQP^O = \QCMA^O$.
\end{proposition}
\begin{proposition}[Corollary 50 of \cite{AIK22}]
    There exists an oracle $O$ relative to which $\P^O = \NP^O$ but $\BQP^O \neq \QCMA^O$.
\end{proposition}
\begin{proposition}[Theorem 29 of \cite{AIK22}]
    There exists an oracle $O$ relative to which $\P^O \neq \NP^O$ (and $\PH^O$ is infinite) but $\P^O \neq \BQP^O = \QCMA^O = \QMA^O$.
\end{proposition}
\noindent Further, it is also interesting to ask how different levels of $\PH$ are related to those of $\QCPH$, for example, can $\NP$ solve some problem which $\BQP$ can not (and vice versa)? The first of these questions was resolved by \cite{BBBV97} and the second by \cite{razOracleSeparationBQP2019}.
\begin{proposition}[Theorem 3.5 of \cite{BBBV97}]
    With probability 1, a random oracle $O$ satisfies $\NP^O \not \subset \BQP^O$.
\end{proposition}
\begin{proposition}[Corollary 1.5 of \cite{razOracleSeparationBQP2019}]
    There exists an oracle $O$ relative to which $\BQP^O \not \subset \PH^O$.
\end{proposition}
\noindent \cite{agkr_qph} also show a quantum-classical variant of the Karp-Lipton theorem for $\QCMA$ which says that $\QCMA$ can not be decided by small non-uniform quantum circuits, unless $\QCPH$ collapses (a weaker version of this was shown by \cite{gharibianQuantumGeneralizationsPolynomial2018}):
\begin{proposition}[Quantum-Classical Karp-Lipton Theorem \cite{agkr_qph}]
    If $\QCMA \in \Bpoly$ then $\QCSigma_2 = \QCPi_2$, where $\Bpoly$ is the class of problems decided by polynomial-sized families of non-uniform quantum circuits.
\end{proposition}
\noindent In terms of classical upper bounds, the best known upper bound on $\PH$ is given by Toda's theorem \cite{T91}, which says that $\PH \subseteq \P^{\PP}$. For $\QCPH$, the best classical upper bound is given by \cite{gharibianQuantumGeneralizationsPolynomial2018}:
\begin{proposition}[Quantum-Classical Toda's Theorem \cite{gharibianQuantumGeneralizationsPolynomial2018}]
    $\QCPH \subseteq \P^{\PP^{\PP}}$
\end{proposition}
\noindent It is not known whether this upper bound can be improved to $\P^{\PP}$. For $\QCPH'$ described previously, however, a $\P^{\PP}$ upper bound can be shown (as noted in \cite{gharibianQuantumGeneralizationsPolynomial2018}): it is easy to see that $\QCSigma_i' \subseteq \Sigma_i^{\BQP}$ and therefore $\QCPH' \subseteq \P^{\PP^{\BQP}}$ because Toda's theorem relativizes. So using the fact that $\BQP$ is low for $\PP$ shown by \cite{FR99}, $\QCPH' \subseteq \P^{\PP}$. A similar argument does not work for $\QCPH$ because it is not known whether $\QCMA \subseteq \NP^{\BQP}$, because if the $\NP$ machine guesses a $\QCMA$ certificate on which the verifier accepts with probability (say) $1/2$, it can no longer query the behaviour of this verifier to the $\BQP$ oracle (which can now return an arbitrary 0/1 response). Finally, note that a much stronger upper bound can be shown for $\QCMA$ (infact, $\QMA$):
\begin{proposition}[Kitaev and Watrous (unpublished), Theorem 3.4 of \cite{MW05}]
    $\BQP \subseteq \QCMA \subseteq \QMA \subseteq \PP$
\end{proposition}
\noindent Now we will discuss how $\QCPH$ is related to quantum variants of $\PH$. \cite{gharibianQuantumGeneralizationsPolynomial2018} defined $\QPH$, where in the $i^{th}$ level $\QSigma_i$, an efficient quantum verifier takes $i$ alternatively quantified mixed state quantum proofs as input:
\begin{align*}
    x\in\ayes \Rightarrow \exists \rho_1 \forall \rho_2 \ldots Q_i \rho_i \colon \Pr[V(x, \rho_1, \rho_2, \ldots, \rho_i)] \geq 2/3 \\
    x\in\ano \Rightarrow \forall \rho_1 \exists \rho_2 \ldots \overline{Q_i} \rho_i \colon \Pr[V(x, \rho_1, \rho_2, \ldots, \rho_i)] \leq 1/3
\end{align*}
Note that the proofs $\rho_k$ are not allowed to be entangled with each other in the definition of $\QPH$ above. This definition clearly generalizes $\QMA$ since $\QSigma_1 = \QMA$. In addition, $\QMA(2) \subseteq \QSigma_3$. However, while mixed state and pure state proofs are known to be equivalent in power for $\QMA$, it is not known to be the case for $\QPH$. \cite{agkr_qph} define $\QPHpure$ with pure state proofs instead of mixed state proofs, which are not allowed to be entangled with each other:
\begin{align*}
    x\in\ayes \Rightarrow \exists \ket{\psi_1} \forall \ket{\psi_2} \ldots Q_i \ket{\psi_i} \colon \Pr[V(x, \ket{\psi_1}, \ket{\psi_2}, \ldots, \ket{\psi_i})] \geq 2/3 \\
    x\in\ano \Rightarrow \forall \ket{\psi_1} \exists \ket{\psi_2} \ldots \overline{Q_i} \ket{\psi_i} \colon \Pr[V(x, \ket{\psi_1}, \ket{\psi_2}, \ldots, \ket{\psi_i})] \leq 1/3
\end{align*}
Once again, $\QSigmapure_1 = \QMA$ and $\QMA(2) \subseteq \QSigmapure_3$. It is easy to see that $\QSigma_i \subseteq \QSigmapure_i$, by doubling the size of each proof for $\QSigma_i$ and sending a purification of the mixed state proof, and therefore $\QPH \subseteq \QPHpure$. The other direction is not known. It is also known that $\QCPH$ is contained in both $\QPH$ and $\QPHpure$ (though it is not known whether $\QCSigma_i \subseteq \QSigma_i$).
\begin{proposition}[Theorem 5.1 of \cite{GrewalY24}]
    $(\forall~i \geq 1)~~~\QCSigma_i \subseteq \QSigma_{2i}$, and therefore $\QCPH \subseteq \QPH$.
\end{proposition}
\begin{proposition}[Lemma 16 of \cite{agkr_qph}]
    $(\forall~i \geq 1)~~~\QCSigma_i \subseteq \QSigmapure_i$, and therefore $\QCPH \subseteq \QPHpure$.
\end{proposition}
\noindent It is known that $\QSigma_2 = \mathsf{QRG}(1)$ (observed by Sanketh Menda, Harumichi Nishimura and John Watrous), where $\mathsf{QRG}(1)$ is the class of one-round zero-sum quantum games. This implies that $\QSigma_2 = \QPi_2 = \mathsf{QRG}(1) \subseteq \PSPACE$. Note however, that no collapse theorem is known for $\QPH$ or $\QPHpure$. Finally, \cite{GrewalY24} define an entangled version of $\QPH$ (which they call $\QEPH$), where the mixed state quantum proofs are allowed to be entangled with each other. They show that $\QEPH$ (even with polynomially many alternating proofs) collapses to its second level, and it is equal to $\mathsf{QRG}(1)$. Summarizing the relationship between fully quantum variants of $\PH$,
\begin{proposition}[Lemma 4.1 of \cite{GrewalY24}, Theorem 5.3 of \cite{agkr_qph}]
    $\QEPH = \mathsf{QRG}(1) = \QSigma_2 \subseteq \QPH \subseteq \QPHpure \subseteq \EXP^{\PP}$
\end{proposition}

\section{Preliminaries}

\begin{proposition}[Chernoff Bound]\label{prop:chernoff}
    Let $\mathbf{X_1}, \cdots, \mathbf{X_n}$ be $n$ independent random variables between $0$ and $1$. If $\mathbf{X} = \sum_{i=1}^n \mathbf{X_i}$ and $\mu = \E[\mathbf{X}]$, then for any $\delta > 0$
    \begin{align*}
        \Pr[\mathbf{X} \geq (1+\delta)\mu] &\leq e^{-\frac{\delta^2}{2+\delta}\cdot\mu} \\
        \Pr[\mathbf{X} \leq (1-\delta)\mu] &\leq e^{-\frac{\delta^2}{2}\cdot\mu}
    \end{align*}
\end{proposition}

\begin{definition}[Restriction]
    Given a function $f:\{0,1\}^n \rightarrow \R$, a restriction $\rho$ is any string in $\{0,1,\ast\}^n$. Let $F = \{i \in [n]|\rho(i) = \ast\}$, then the restricted function $f_\rho: \{0,1\}^F \rightarrow \R$ is defined as $f_\rho(x) = f(y)$ where
    \begin{align*}
        y_i = \begin{cases}
            x_i & \text{if $i \in F$} \\
            \rho(i) & \text{otherwise}
        \end{cases}
    \end{align*}
    $\rho$ is a random restriction, if it is sampled from some probability distribution $D$ on $\{0,1,\ast\}^n$.
\end{definition}

\begin{definition}[Projection]
    Given a function $f:\{0,1\}^{n \times l} \rightarrow \R$, and a restriction $\rho$ in $\{0,1,\ast\}^{n \times l}$, the projected function $\mathsf{proj}_{\rho}f: \{0,1\}^n \rightarrow \R$ is defined as $\mathsf{proj}_{\rho}f(x) = f(y)$ where
    \begin{align*}
        y_{i,j} = \begin{cases}
            x_i & \text{if $\rho(i,j) = \ast$} \\
            \rho(i,j) & \text{otherwise}
        \end{cases}
    \end{align*}
    So the projection operator maps all unrestricted variables in a given block of size $l$, to the same new variable. If $\rho$ is a random restriction, then $\mathsf{proj}_{\rho}$ is a random projection.
\end{definition}

\subsection{Quantum-Classical Polynomial Hierarchy}
\begin{definition}[$\QCSigma_i$]\label{def:QCSigmam}
  Let $A=(\ayes,\ano)$ be a promise problem. We say that $A$ is in $\QCSigma_i(c, s)$ for poly-time computable functions $c, s: \mathbb{N} \mapsto [0, 1]$ if there exists a poly-bounded function $p:\mathbb{N}\mapsto\mathbb{N}$ and a poly-time uniform family of quantum circuits $\{V_n\}_{n \in \mathbb{N}}$ such that for every $n$-bit input $x$, $V_n$ takes in classical proofs ${y_1}\in \set{0,1}^{p(n)}, \ldots, {y_i}\in \set{0,1}^{p(n)}$ and outputs a single qubit, {such that:}
  \begin{itemize}
    \item Completeness: $x\in \ayes$ $\Rightarrow$ $\exists y_1 \forall y_2 \ldots Q_i y_i$ s.t. $\Pr[V_n \text{ accepts } (y_1, \ldots, y_i)] \geq c$.
    \item Soundness: $x\in \ano$ $\Rightarrow$ $\forall y_1 \exists y_2 \ldots \overline{Q}_i y_i$ s.t. $\Pr[V_n \text{ accepts } (y_1, \ldots, y_i)] \leq s$.
  \end{itemize}
  Here, $Q_i$ equals $\exists$ when $i$ is odd and equals $\forall$ otherwise and $\overline{Q}_i$ is the complementary quantifier to $Q_i$. Define
  \begin{align*}
  \QCSigma_i := \bigcup_{{ c - s \in \Omega(1/\poly(n))}} \QCSigma_i(c, s).
  \end{align*}
\end{definition}

\noindent Note that the first level of this hierarchy corresponds to $\QCMA$. The complement of the $i^{\mathrm{th}}$ level of the hierarchy, $\QCSigma_i$, is the class $\QCPi_i$ defined next.

\begin{definition}[$\QCPi_i$]\label{def:QCPim}
  Let $A=(\ayes,\ano)$ be a promise problem. We say that $A \in \QCPi_i(c, s)$ for poly-time computable functions $c, s: \mathbb{N} \mapsto [0, 1]$ if there exists a polynomially bounded function $p:\mathbb{N}\mapsto\mathbb{N}$ and a poly-time uniform family of quantum circuits $\{V_n\}_{n \in \mathbb{N}}$ such that for every $n$-bit input $x$, $V_n$ takes in classical proofs ${y_1}\in \set{0,1}^{p(n)}, \ldots, {y_i}\in \set{0,1}^{p(n)}$ and outputs a single qubit, {such that:}
  \begin{itemize}
    \item Completeness: $x\in \ayes$ $\Rightarrow$ $\forall y_1 \exists y_2 \ldots Q_i y_i$ s.t. $\Pr[V_n \text{ accepts } (y_1, \ldots, y_i)] \geq c$.
    \item Soundness: $x\in \ano$ $\Rightarrow$ $\exists y_1 \forall y_2 \ldots \overline{Q}_i y_i$ s.t. s.t. $\Pr[V_n \text{ accepts } (y_1, \ldots, y_i)] \leq s$.
  \end{itemize}
  Here, $Q_i$ equals $\forall$ when $i$ is odd and equals $\exists$ otherwise, and $\overline{Q}_i$ is the complementary quantifier to $Q_i$. Define
  \begin{align*}
  \QCPi_i := \bigcup_{{c - s \in \Omega(1/\poly(n))}} \QCPi_i(c, s).
  \end{align*}
\end{definition}

\noindent Now the corresponding quantum-classical polynomial hierarchy is defined as follows.

\begin{definition}[Quantum-Classical Polynomial Hierarchy ($\QCPH$)]\label{def:QCPH}
  \begin{align*}
        \QCPH = \bigcup_{i \in \mathbb{N}} \; \QCSigma_i = \bigcup_{i \in \mathbb{N}} \; \QCPi_i.
    \end{align*}
\end{definition}

\begin{definition}[$\QAC^0_i(s)$ circuits]
    A circuit family $\{C_n\}_{n=1}^\infty$ is in $\QAC^0_i(s)$ for some function $s: \bN \rightarrow \bN$ if it satisfies the following:
    \begin{enumerate}
        \item $C_n$ has $i$ alternating layers of $\vee$ and $\wedge$ gates, followed by a layer of quantum query circuits of size $\polylog(s(n))~\forall n \in \bN$
        \item size($C_n$) $\leq s(n)~\forall n \in \bN$ where size of the circuit is the number of $\vee$, $\wedge$ gates plus the number of quantum query circuits in the bottom layer
    \end{enumerate}
\end{definition}
Note that in the above definition, if the quantum query circuits are not computing total functions, then we treat each gate in the circuit as computing a partial function, whose output is evaluated in the following manner. For an input $x$ to the circuit, a $\vee$ gate outputs 1 if at least one of its inputs is 1 regardless of the arbitrary 0/1 responses of the input partial functions whose promise is violated. It outputs 0 if all of its inputs are 0 regardless of the arbitrary 0/1 responses of the input partial functions whose promise is violated. Otherwise, the output of the $\vee$ gate is undefined. The evaluation of the output of $\wedge$ gate is done in a similar manner. This method of evaluation gives the same output as the one if we evaluate the output of the circuit in a manner consistent with \defn{QCSigmam} (or \defn{QCPim}). That is, a $\vee$ gate corresponds to an $\exists$ quantifier and a $\wedge$ gate corresponds to a $\forall$ quantifier. We show this in \lem{circuitconsistent}.

\begin{lemma}\label{lem:circuitconsistent}
    The method of circuit evaluation discussed above gives the same output as the one if we evaluate the output of the circuit in a manner consistent with \defn{QCSigmam} (or \defn{QCPim}).
\end{lemma}
\begin{proof}
    (Sketch) For $\QCMA$, since there is only one OR gate followed by a layer of quantum query algorithms in the circuit, the two notions both assign a value to this top OR gate. For the second level $\QCSigma_2$, the circuit consists of a top OR gate, which takes inputs from a layer of AND gates, which take inputs from a layer of quantum query algorithms. In this case, if we want to assign a value to the top gate of the circuit according to \defn{QCSigmam}, in the “yes” case, we assign a 1 if there is a 1 input to the OR gate (existential quantifier), which in turn means that there is an AND gate (universal quantifier) set to 1, which happens if this AND gate receives all inputs as 1, regardless of the arbitrary responses chosen for the quantum query algorithms (partial functions) whose promise is violated. Looking at this evaluation in a bottom-up manner, this AND gate is receiving all inputs as 1, regardless of the arbitrary partial function responses, which means the output of this AND gate is well-defined and set to 1. For the top layer OR gate it means that it is getting a well-defined 1 input from this specific AND gate, which means it will evaluate to 1. We can also assign 0/1/undefined values to the other AND gates in the middle layer, and the output of the OR gate remains 1 since we have already shown one AND gate which has well-defined output 1. A similar argument can be made for the “no” case, and for the higher levels of the hierarchy. \\
    (Formal Argument) We use induction on the depth of the circuit $\QAC^0_d$ to establish the claim.
    \begin{itemize}
        \item Base Case $(d=1)$: For $\QCMA$, since there is only one $\vee$ gate followed by a layer of quantum query algorithms in the circuit, the two notions both assign a value to this top $\vee$ gate, depending on whether there is a quantum algorithm which outputs $1$. Similar argument holds for $\QCPi_1$.
        \item Induction Hypothesis: Suppose the claim is true for $\QAC^0_{d-1}$.
        \item Induction Step: We think of a $\QAC^0_d$ circuit as a $\vee$ gate, taking inputs from multiple $\QAC^0_{d-1}$ circuits with a top $\wedge$ gate (all of which get a common input $x$). For this input $x$, from the induction hypothesis, the responses of the circuits are consistent with $\QCPi_{d-1}$ predicates. In the yes case, there exists a proof $y_1$ such that the resulting predicate on fixing $y_1$ is a yes instance of a $\QCPi_{d-1}$ predicate. In this case, the corresponding $\QAC^0_{d-1}$ circuit has a well-defined output $1$ from the induction hypothesis. Note that the other $\QAC^0_{d-1}$ circuits which are input to the top $\vee$ gate might not have well-defined outputs, but the output of the $\vee$ gate will be $1$ because it is always getting a $1$ input (on the fixed $y_1$ branch). In the no case, for any proof $y_1$, the resulting predicate on fixing $y_1$ is a no instance of a $\QCPi_{d-1}$ predicate. In this case, all the corresponding $\QAC^0_{d-1}$ circuits have a well-defined output $0$ from the induction hypothesis. Therefore, the output of the $\vee$ gate will be $0$ because it is always getting $0$ inputs on all its branches.
    \end{itemize}
    Therefore the claim follows.
\end{proof}

\subsection{Query complexity}
In query complexity, an algorithm (classical or quantum) is given black-box access to a string $x \in \{0,1\}^N$, and is supposed to compute some boolean function $f(x)$ by querying $x$ as few times as possible. Interested readers can refer to \cite{BdW02} for a detailed introduction to query complexity. Let $N = 2^n$. In the classical query model, an algorithm queries $i \in \{0,1\}^n$ and receives $x_i$, which we will denote by $O_x(i) = x_i$. In the quantum query model, an algorithm queries $\sum_{i \in \{0,1\}^n} \alpha_i \ket{i}$ and receives $\sum_{i \in \{0,1\}^n} \alpha_i (-1)^{x_i} \ket{i}$, which we will denote by $O_x(\sum_{i \in \{0,1\}^n} \alpha_i \ket{i}) = \sum_{i \in \{0,1\}^n} \alpha_i (-1)^{x_i} \ket{i}$.
\begin{definition}[Decision Tree]
    A classical decision tree $T$ computing a (partial) function $f:\{0,1\}^N \rightarrow\{0,1,\perp\}$ (where $N = 2^n$) with $D$ queries is a classical query algorithm defined as follows:
    \begin{itemize}
        \item If $D = 0$, then $T(x)$ is equal to $0$ or $1$.
        \item Else, $T$ queries some $i \in \{0,1\}^n$ and $T(x) = T_{x_i}(x)$ where $T_{x_i}$ is a decision tree of depth $D-1$.
        \item For every $x\in f^{-1}(\{0,1\})$, $T(x) = f(x)$.
    \end{itemize}
\end{definition}

\begin{definition}[Quantum Query Algorithm]
    A quantum query algorithm $Q$ computing a (partial) function $f:\{0,1\}^N \rightarrow\{0,1,\perp\}$ (where $N = 2^n$) with $T$ queries is defined as follows:
    \begin{align*}
        Q(x) = U_{T+1}O_xU_{T}\ldots O_x U_1\ket{0^n}\ket{0^w}
    \end{align*}
    where $\ket{0^w}$ is the workspace, $O_x$ is the query oracle defined earlier and $U_i$ are unitaries. This is followed by a measurement of the first qubit, such that for every $x \in \{0,1\}^N$ if $f(x) = 1$ then $Q(x)$ outputs $1$ with probability at least $2/3$ and if $f(x) = 0$ then $Q(x)$ outputs $1$ with probability at most $1/3$.
\end{definition}

\begin{definition}[Query Magnitudes \cite{BBBV97}]
    Given a quantum query algorithm $Q$ computing a (partial) function $f:\{0,1\}^N \rightarrow\{0,1,\perp\}$ (where $N = 2^n$) with $T$ queries, define query magnitude at time step $t$, denoted $m_i^t(x)$, for $t \in [T]$ and $i \in \{0,1\}^n$ as follows:
    \begin{align*}
        Q(x) &= U_{T+1}O_xU_{T}\ldots O_x U_1\ket{0^n}\ket{0^w} \\
        m_i^t(x) &= \Pr[\text{Measuring query register of } U_{t+1}O_xU_{t}\ldots O_x U_1\ket{0^n}\ket{0^w} \text{ gives }i] \\
        m_i(x) &= \sum_{t=1}^T m_i^t(x)
    \end{align*}
    We call $m_i(x)$ as the query magnitude of $Q$ for bit $i$ on input $x$. Since $\sum_{i \in \{0,1\}^n} m_i^t(x) = 1$ for all $t \in [T]$, we know that $\sum_{i \in \{0,1\}^n} m_i(x) \leq T$ (it need not be exactly equal since the quantum algorithm might make less than $T$ queries on some inputs $x$).
\end{definition}

\begin{definition}[Partial Assignment]
    A partial assignment $p$ of length $N$ is any string from the set $\{0,1,\ast\}^N$. A string $x \in \{0,1,\ast\}^N$ is consistent with $p$ if $p_i = x_i$ whenever $p_i \neq \ast$.
\end{definition}

\begin{definition}[Certificate]
    Given a (partial) function $f:\{0,1\}^N \rightarrow\{0,1,\perp\}$, a partial assignment $p \in \{0,1,\ast\}^N$ is a $b$-certificate for $f$ if $f(x) = f(y) = b$ for all $x, y$ in the domain of $f$ which are consistent with $p$. The length of this certificate $p$ is the number of bits in $p$ not equal to $\ast$.
\end{definition}

\begin{definition}[Certificate Complexity]
    Given a (partial) function $f:\{0,1\}^N \rightarrow\{0,1,\perp\}$, the certificate complexity of $f$ at input $x \in \mathsf{dom}(f)$ is defined as $\C(f,x) = \min_p \mathsf{len}(p)$ where the minimization is over all partial assignments $p$ consistent with $x$, which are also $f(x)$-certificates for $f$. Then the certificate complexity $\C(f) = \max_{x \in \mathsf{dom}(f)} \C(f,x)$.
\end{definition}

\begin{definition}[Sensitivity]
    Given a (partial) function $f:\{0,1\}^N \rightarrow\{0,1,\perp\}$, the sensitivity of $f$ at input $x \in \mathsf{dom}(f)$ is defined as $\s(f,x) = |\{i \in \{0,1\}^n : x^{\oplus i} \in \mathsf{dom}(f), f(x) \neq f(x^{\oplus i})\}|$, where $x^{\oplus i}$ is the string $x$ with the $i^{th}$ bit flipped. Then the sensitivity $\s(f) = \max_{x \in \mathsf{dom}(f)} \s(f,x)$.
\end{definition}

\begin{definition}[Block Sensitivity]
    Given a (partial) function $f:\{0,1\}^N \rightarrow\{0,1,\perp\}$, a block $B \subseteq [N]$ is a sensitive block for $f$ at input $x \in \mathsf{dom}(f)$ if $x^B \in \mathsf{dom}(f)$ and $f(x) \neq f(x^B)\}$, where $x^B$ is the string $x$ with the bits in $B$ flipped. If $B$ is a sensitive block, we say $s(f,x,B) = 1$ and $0$ otherwise. Then the block sensitivity of $f$ at input $x \in \mathsf{dom}(f)$, denoted $\bs(f,x)$ is the maximum number of disjoint sensitive blocks of $f$ at input $x$. Then the block sensitivity $\bs(f) = \max_{x \in \mathsf{dom}(f)} \bs(f,x)$.
\end{definition}

\begin{definition}[Fractional Certificate Complexity]
    Given a (partial) function $f:\{0,1\}^N \rightarrow\{0,1,\perp\}$, the set $c_x = \{c_{x,i} \geq 0: i \in \{0,1\}^n\}$ is a fractional certificate for $f$ at input $x \in \mathsf{dom}(f)$ if for every $y \in \mathsf{dom}(f)$, $\sum_{i: x_i \neq y_i} c_{x,i} \geq |f(x) - f(y)|$. The fractional certificate complexity of $f$ at input $x \in \mathsf{dom}(f)$, denoted $\mathsf{FC}(f,x) = \min_{c_x} \sum_i c_{x,i}$ where the minimization is done over all fractional certificates $c_x$ for $f$ at input $x$. Then the fractional certificate complexity $\mathsf{FC}(f) = \max_{x \in \mathsf{dom}(f)} \mathsf{FC}(f,x)$.
\end{definition}

\begin{definition}[Fractional Block Sensitivity]
    Given a (partial) function $f:\{0,1\}^N \rightarrow\{0,1,\perp\}$, the set $w_x = \{w_{x,B} \geq 0: B \subseteq [N]\}$ is a fractional block for $f$ at input $x \in \mathsf{dom}(f)$ if for every $i \in \{0,1\}^n$ we have that $\sum_{B: i \in B} w_{x,B}.s(f,x,B) \leq 1$. The fractional block sensitivity of $f$ at input $x \in \mathsf{dom}(f)$, denoted $\mathsf{fbs}(f,x) = \max_{w_x} \sum_B w_{x,B}.s(f,x,B)$ where the maximization is done over all fractional blocks $w_x$ for $f$ at input $x$. Then the fractional block sensitivity $\mathsf{fbs}(f) = \max_{x \in \mathsf{dom}(f)} \mathsf{fbs}(f,x)$.
\end{definition}

\begin{definition}[Approximate Degree]
    Given a (partial) function $f:\{0,1\}^N \rightarrow\{0,1,\perp\}$, the approximate degree of $f$, denoted $\widetilde{\mathsf{deg}}(f)$ is the minimum degree of a multi-linear polynomial $p$ which approximate $f$ to error within additive error $1/3$ for all input $x \in \mathsf{dom}(f)$ and $p(x) \in [0,1]$ for all $x \in \{0,1\}^N$.
\end{definition}

\begin{lemma}[\cite{BBC+01}]\label{lem:querytodeg}
    If a (partial) function $f:\{0,1\}^N \rightarrow\{0,1,\perp\}$ is computable by a quantum algorithm $Q$ making $T$ queries, we have that there is a multilinear polynomial $p$ of degree at most $2T$ such that $p(x)$ is equal to the probability that $Q$ outputs 1 on input $x$. Therefore $p$ approximates $f$ to error within additive error $1/3$ for all input $x \in \mathsf{dom}(f)$ and $p(x) \in [0,1]$ for all $x \in \{0,1\}^N$ and $\widetilde{\mathsf{deg}}(f) \leq 2T$.
\end{lemma}

\subsection{\texorpdfstring{$\QCMA$}{QCMA} Hierarchy}
In this section, we define the $\QCMA$ hierarchy (called $\QCMAH$). It is a hierarchy of classes, where the $i^{th}$ level, called $\QCMAH_i$ is defined as follows:
\begin{align*}
    \QCMAH_1 &:= \QCMA \\
    \QCMAH_i &:= \QCMA^{\QCMAH_{i-1}} \\
    \QCMAH &:= \cup_{i=1}^{\infty} \QCMAH_i
\end{align*}
Note that $\QCMA$ is a class of promise problems, thus when querying a $\QCMAH_i$ oracle, any $\QCMA$ algorithm $V$ is making queries to a promise oracle. Note that the $\QCMA$ \emph{algorithm} $V$ is different from the $\QCMA$ \emph{verifier} (which takes a single proof as input and outputs 0/1), here $V$ instead refers to the OR of all these 0/1 outputs of the verifier for every possible proof. If $V$ makes a query on an input which violates the promise of the $\QCMAH_i$ oracle, then the oracle can return an arbitrary 0/1 response. Therefore we need to be careful about when the behaviour of $V$ is well-defined. In this paper, we follow the definition of well-defined behaviour proposed in \cite{AIK22} (note however that there can be many other reasonable definitions of well-defined behaviour, though we do not describe them here). The behaviour of $V$ is well-defined, if the output of $V$ remains the same regardless of the arbitrary 0/1 responses of the $\QCMAH_i$ oracle on inputs violating the promise. In particular for a yes input, the $\QCMA$ verifier might accept only one proof for some choice of oracle responses outside the promise, and it might accept (say) half of the proofs for some other choice of oracle responses, but the behaviour remains well-defined as long as it accepts at least one proof for every choice of oracle responses. \\
In the rest of the paper, we will denote as $\QCMAH_i$ verifier $V$ as an $i$-tuple of $\QCMA$ verifiers $V = \langle V_1, \ldots, V_i \rangle$ where $V_1$ corresponds to the base $\QCMA$ verifier and the remaining $i-1$ verifiers form the $(i-1)$-tuple for the verifier corresponding to the $\QCMAH_{i-1}$ oracle.

\section{Polynomials and fractional block sensitivity}

In this section, we establish a relationship between fractional
block sensitivity and degree for real-valued functions. We will need
this relationship to establish a switching lemma for approximate
degree (the tools of this section are not necessary for the quantum
switching lemma, which may instead rely on only the hybrid argument
of \cite{BBBV97}, but we opt to prove the stronger switching lemma
for polynomials).

\subsection{Definitions for real-valued functions}

When studying the degree of bounded real-valued functions, most of
the literature uses the convention that the functions have signature
$\sB^n\to[-1,1]$, known as the $\pm 1$ basis. Switching between
$\B$ and $\sB$ does not affect most measures: we can interpret
$-1$ as $1$ and $+1$ as $0$, and we can plug in $1-2x_i$ in
each variable to convert a $\B$ variable to a $\sB$ variable,
or plug in $(1-x_i)/2$ to go in the reverse direction. This
does not affect the degree. We will use the $\sB$ basis in this
section.

\paragraph{Sensitivity.}
The sensitivity of a real-valued function can be defined as follows.
For a block $B\subseteq[n]$, define
\[\s(f,x,B)\coloneqq \frac{|f(x)-f(x^B)|}{2},\]
where the notation $x^B$ refers to the string $x$
with the block $B\subseteq[n]$ of bits flipped.
Note that we divide by $2$ because $f$ can take values from $[-1,1]$
instead of $[0,1]$.
This is the sensitivity of a specific block
$B$ to a specific input $x\in\sB^n$, with respect to the function $f$;
it is a value between $0$ and $1$. We next define the sensitivity
of $f$ at $x$ to be
\[\s(f,x)\coloneqq \sum_{i=1}^n \s(f,x,\{i\}),\]
which is the total sensitivity of all the bits of $x$. We then
define
\[\s(f)\coloneqq \max_{x\in\sB^n} \s(f, x).\]
For a Boolean function, this definition matches the usual
definition of sensitivity $\s(f)$.

\paragraph{Block sensitivity.}
We define block sensitivity analogously:
$\bs(f,x)$ will be the maximum possible value of
$\sum_{j=1}^k \s(f,x,B_j)$ over choices of disjoint blocks
$B_j\subseteq[n]$, and $\bs(f)$ will be the maximum possible value
of $\bs(f,x)$ over choices of inputs $x\in\sB^n$.

\paragraph{Fractional block sensitivity.}
We define $\fbs(f,x)$ to be the maximum possible value of the sum
$\sum_{B\subseteq[n]} w_B\s(f,x,B)$, subject to the constraints
$w_B\ge 0$ for all $B\subseteq[n]$ and $\sum_{B\ni i} w_B\le 1$
for all $i\in[n]$. (The weights $w_B$ represent fractions of a block,
and the total weight of all blocks containing a bit $i$ must
be at most $1$, making the blocks ``fractionally disjoint''.)
As usual, we define $\fbs(f)\coloneqq \max_{x\in\sB^n}\fbs(f,x)$.

\paragraph{Fractional certificate complexity.}
The definition of $\fbs(f,x)$ is the optimal value of
a linear maximization program in terms of the real variables $w_B$.
The dual of this linear program can is as follows.
The variables are $c_i\ge 0$, with the constraint
$\sum_{i\in B} c_i\ge \s(f,x,B)$ for all $B\subseteq[n]$.
The objective is to minimize $\sum_{i\in[n]} c_i$.
This can be interpreted as finding a fractional certificate,
with $c_i$ representing the fraction with which bit $i$ is used
in the certificate, such that any block $B$ which, when flipped,
changes the value of the function by $\s(f,x,B)$ must be
``detected'' by the certificate in the sense that the certificate
puts total weight at least $\s(f,x,B)$ on the bits of $B$.

Fractional certificate complexity is equal to fractional block
sensitivity, so we will not give it new notation. However,
we denoting by $c_{x,i}$ the fractional certificate for $x$,
we get that for all $x,y$,
\[\sum_{i:x_i\ne y_i} c_{x,i}\ge \frac{|f(x)-f(y)|}{2},
\qquad \sum_{i=1}^n c_{x,i} \le \fbs(f).\]
The factor of $2$ comes from the use of $[-1,1]$ outputs for $f$;
if $f$ takes outputs in $[0,1]$ instead (for example,
if we consider the function $f(x)=(1-g(x))/2$ where
$g$ takes values in $[-1,1]$ we are comparing the differences
$|f(x)-f(y)|$ to the value of $\fbs(g)$),
we do not need to divide by $2$.

\subsection{Relationships to degree}

Our results will rely on the following theorem by \cite{FHKL16},
which states that sensitivity lower bounds the squared
degree even for a real-valued function. This theorem follows from
results in approximation theory (though it is possible to get
a version which is weaker by a constant factor using a proof
analogous to the one for the discrete case).

\begin{theorem}[\cite{FHKL16}]\label{thm:sdeg}
For any $f\colon\sB^n\to[-1,1]$, we have $\s(f)\le \deg(f)^2$.
\end{theorem}

With this theorem in hand, we show that block sensitivity
also lower bounds the squared degree.

\begin{corollary}\label{cor:bsdeg}
For any $f\colon\sB^n\to[-1,1]$, we have $\bs(f)\le \deg(f)^2$.
\end{corollary}

\begin{proof}
This follows easily from \thm{sdeg}. Fix a function $f$, and let
$x\in\sB^n$ be such that $\bs(f)=\bs^{(1)}(f,x)$. Let
$B_1\sqcup\dots \sqcup B_k$ be a partition of $[n]$ into blocks which
maximizes the total sensitivity, so that
$\bs(f)=\sum_{j=1}^k \s^{(1)}(f,x,B_j)$.
Define a new function $f'\colon\sB^k\to[-1,1]$ as follows.
For any $y\in\sB^k$, let $B_y=\bigsqcup_{j:y_j=-1} B_j$, and define
$f'(y)=f(x^{B_y})$. It is easy to check that the sensitivity
of the input $1^k$ to $f'$ is the same as the block sensitivity of
the input $x$ to $f$. Therefore, we have
$\bs(f)\le \s(f')\le \deg(f')^2$. Now, we can represent $f'$ as a
polynomial by taking the polynomial for $f$ and plugging in
$(x_1y_{j(1)},x_2 y_{j(2)},\dots, x_ny_{j(n)})$, where $j(i)$ is
the unique $j\in[k]$ such that $i\in B_j$. Treating $x_i$ as
constants and $y_j$ as the variables, this polynomial for $f'$ has
degree at most $\deg(f)$, so we conclude $\bs(f)\le \deg(f)^2$,
as desired.
\end{proof}

Finally, we show that fractional block sensitivity lower bounds
the squared degree.

\begin{theorem}\label{thm:fbstodeg}
For any $f\colon\sB^n\to[-1,1]$, we have
$\fbs(f)\le \frac{\pi^2}{4}\deg(f)^2$.
\end{theorem}

\begin{proof}
Our proof follows the strategy of \cite{ABK21}, which showed a
similar result for Boolean functions; they converted fractional
block sensitivity to non-fractional block sensitivity by composing
the function with large copies of the OR function.

We proceed as follows. Let $x$ maximize $\fbs(f,x)$
and let $\{w_B\}_B$ be an optimal weight scheme, so that
$\fbs(f)=\sum_B w_B\s^{(1)}(f,x,B)$.
Pick $N\in\bN$ to be a large integer (we will take the limit
as $N\to\infty$ at the end of the proof).
Round down the weights $w_B$ so that they
are rationals with denominator $N$; the rounding down ensures
the weight scheme stays feasible. The total weight is now some
$T\le \fbs(f)$, but as $N$ goes to infinity, $T$ can be made
arbitrarily close to $\fbs(f)$. Therefore, it suffices to show
$T\le\pi^2\deg(f)^2$, where $T$ is the objective value of a weight
scheme in which all weights are rational numbers with denominator $N$.

Let $m$ be the integer between $\pi\sqrt{N}/2$ and $\pi\sqrt{N}/2+1$.
Consider the function $h\colon\sB^N\to\bR$ defined by
\[h(y)=\alpha+\frac{1-\alpha}{N}\sum_{j=1}^N y_j,\]
where $\alpha=1-(N/2)(1-\cos(\pi/m))$.
Note that
$\cos x\ge 1-x^2/2$, so $\alpha\ge 1-(N/2)(\pi^2/2m^2)\ge 0$.
Hence $\alpha> 0$, and we also have $\alpha< 1$. From this it follows
that the range of $h$ is within $[-1,1]$, that is, $h$ is bounded.
Also, $h(1^N)=1$ and $h((1^N)^j)=\cos(\pi/m)$, where $(1^N)^j$ denotes
the string $1^N$ with index $j$ flipped to $-1$.

Let $T_m$ be the Chebyshev polynomial of degree $m$. This is a
single-variate real polynomial mapping $[-1,1]$ to $[-1,1]$,
satisfying $T_m(1)=1$ and $T_m(\cos(\pi/m))=-1$. Consider the
polynomial $T_m\circ h$; this is a bounded polynomial of degree $m$
in $N$ variables, and it satisfies $T_m\circ h(1^N)=1$ and
$T_m\circ h((1^N)^j)=-1$.

For $z\in\sB^{nN}$, we write $z=z_1z_2\dots z_n$ where $z_i\in\sB^N$.
Define the function $f'\colon\sB^{nN}\to[-1,1]$ by
\[f'(z)=f(x_1T_m(h(z_1)),x_2T_m(h(z_2)),\dots,x_n T_m(h(z_n))),\]
where on non-Boolean inputs, the function $f$ is defined by its
unique polynomial (which is in turn defined by the behavior of $f$
on $\sB^n$).
Recall that $x\in\sB^n$ is the input to $f$ which maximized $\fbs$,
which is a constant; the only variables are $z$.
Note that $\deg(f')\le m\deg(f)$, because we are plugging in
polynomials of degree $m$ into a polynomial of degree $\deg(f)$.

Observe that $f'(1^{nN})=f(x)$.
Now, with each block $B\subseteq[n]$, we associate $Nw_B$ blocks
of $f'$ (recall that $Nw_B$ is a non-negative integer since we
rounded the weights $w_B$). Moreover, we will ensure all the blocks
of $f'$ are disjoint. Given a block $B\subseteq[n]$ of $f$,
we can get one block $B'$ of $f'$ by picking a variable inside
the $N$-bit string $z_i$ for each $i\in B$ (that is, $B'$ will
contain $|B|$ variables in total). Observe that
$f'((1^{nN})^{B'})=f(x^B)$, since flipping a single bit of $z_i$
flips $T_m(h(z_i))$. Therefore, flipping such a block $B'$ will
behave the same with respect to $f'$ (at input $1^{nN}$)
as flipping $B$ behaved with respect to $f$ (at input $x$).

As mentioned, we pick $Nw_B$ blocks $B'$ for each block $B$ of $f$.
To ensure all these blocks are disjoint (even for different values
of $B$), we pick different bits inside the $N$-bit strings $z_i$
each time such a bit is required. Since the weights $w_B$
are feasible, we know that $\sum_{B\ni i} w_B\le 1$, so
$\sum_{B\ni i} Nw_B\le N$; therefore, the $N$ bits in $z_i$ suffice
to ensure that each block gets a unique bit inside each $z_i$.

We have now defined a set of disjoint blocks for $f'$.
The total sensitivity of all these blocks with respect to the input
$1^{nN}$ is exactly $\sum_B Nw_B\s^{(1)}(f,x,B)=NT$, where $T$
was the objective value of our set of weights ($T$ is arbitrarily
close to $\fbs(f)$). We therefore have
$\bs(f')\ge NT$. Using \cor{bsdeg}, we get
\[NT\le \bs(f')\le \deg(f')^2\le m^2\deg(f)^2
\le (\pi\sqrt{N}/2+1)^2\deg(f)^2.\]
Dividing both sides by $N$, we get
\[T\le (\pi^2/4+O(1/\sqrt{N}))\deg(f)^2.\]
Since $T$ can be made arbitrarily close to $\fbs(f)$ as $N$ gets
arbitrarily large, we must have $\fbs(f)\le (\pi^2/4)\deg(f)^2$,
as desired.
\end{proof}

\thm{fbsdeg} follows.

\section{Random restrictions for polynomials}
In this section, we show that functions which have low fractional block sensitivity (and therefore, low approximate degree and low quantum query complexity) become \emph{close} to a small-depth decision tree with high probability, after applying a sufficiently strong random restriction. This resolves an open problem in \cite{AIK22}, asking whether functions with low approximate degree become close to a DNF on applying a random restriction. We first define a notion of $f(x)$-certificates for a function $F:\{0,1\}^N \rightarrow \{0,1,\perp\}$, which is approximated to error 1/3 by a real-valued function $f$.

\begin{definition}[$f(x)$-certificate]\label{def:certificate}
    Let $F:\{0,1\}^N \rightarrow \{0,1,\perp\}$ be a partial function which is approximated to error 1/3 by a real-valued function $f$. Then a pair $(K, x)$ for $K \subseteq [N]$ is a 1-certificate for $F$ on input $x \in \{0,1\}^N$ if for all $y \in \{0,1\}^N$ such that $y_i = x_i$ for all $i \in K$, $f(y) > 1/2$. In particular, the bits of $x$ in $K$ certify that $F(x) \neq 0$. Similarly, a pair $(K, x)$ for $K \subseteq [N]$ is a 0-certificate for $f$ on input $x \in \{0,1\}^N$ if for all $y \in \{0,1\}^N$ such that $y_i = x_i$ for all $i \in K$, $f(y) \leq 1/2$.
\end{definition}

Next, we state a random restriction result shown by \cite{AIK22} for quantum query algorithms (with minor changes in parameters), which they use to show that efficient QMA-query algorithms become \emph{close} to a DNF in expectation over the choice of a suitably strong random restriction.

\begin{theorem}[Theorem 60 of \cite{AIK22}]\label{thm:randrestbqp}
    Let $f:\{0,1\}^N \rightarrow \{0,1,\perp\}$ be a partial function computable by a quantum query algorithm $Q$ making $T$ queries to the input. Pick a set $S \subseteq [N]$ where each index $i \in [N]$ is put in $S$ independently with probability $p$. Choose an arbitrary $k \in \mathbb{N}$ and $x \in \{0,1\}^N$. Define $\tau = \frac{2p^{3/4}T}{k}$ and $K = \{i \in S: m_i(x) > \tau\}$, where $m_i(x)$ is the query magnitude of $Q$ on index $i$ when the input is $x$. Then with probability at least $1-2e^{-k/6}$ over choice of $S$, $|K| \leq k$, $|S| \leq 2pN$ and for all $y \in \{0,1\}^N$ with $\{i \in [N]: x_i \neq y_i\} \subseteq S \setminus K$, we have:
    \begin{align*}
        |\Pr[Q(x) = 1] - \Pr[Q(y) = 1]| \leq 16Tp^{7/8}\sqrt{N/k}.
    \end{align*}
\end{theorem}
Note further in \thm{randrestbqp} above that $|K| \geq 1$ with probability at most $\frac{k}{2}p^{1/4}$ by a union bound. In particular, if $k = \polylog(N)$ and $p$ is sufficiently small (say $1/N^{3/4}$), then with reasonably high probability over the choice of $S$ (say $1/N^{2/3}$), none of the indices $i$ such that $m_i(x) > \tau$ are included in $S$. Therefore, the restricted function becomes a constant function in this case.

We will need to reprove \thm{randrestbqp} for functions with low fractional block sensitivity. The proof is analogous to that of \cite{AIK22}, with query magnitudes of a quantum algorithm replaced by the weight assignment for fractional block sensitivity.

\begin{lemma}[Analogue of Theorem 60 of \cite{AIK22}]\label{lem:restlowdeg}
    Let $f:\{0,1\}^N \rightarrow \{0,1,\perp\}$ be a partial function which is approximated to error 1/3 by a real-valued function $\Tilde{f}$ of fractional block-sensitivity $F$. Pick a set $S \subseteq [N]$ where each index $i \in [N]$ is put in $S$ independently with probability $p$. Choose an arbitrary $k \in \mathbb{N}$ and $x \in \{0,1\}^N$. Define $\tau = \frac{2p^{3/4}F}{k}$ and $K = \{i \in S: c_{x,i} > \tau\}$, $c_{x,i}$ is the fractional certificate of $\Tilde{f}$ on index $i$ when the input is $x$. Then with probability at least $1-2e^{-k/6}$ over choice of $S$, $|K| \leq k$, $|S| \leq 2pN$ and for all $y \in \{0,1\}^N$ with $\{i \in [N]: x_i \neq y_i\} \subseteq S \setminus K$, we have:
    \begin{align*}
        |\Tilde{f}(x) - \Tilde{f}(y)| \leq \frac{4 p^{7/4} F N}{k}
    \end{align*}
\end{lemma}
\begin{proof}
    Define $H = \{i \in [N]: c_{x,i} > \tau\}$, so that $K = H \cap S$. Note that since $\sum_{i=1}^n c_{x,i} \leq F$, so $|H| \leq \frac{k}{2p^{3/4}} \leq \frac{k}{2p}$. Since $K \subseteq H$, we know that $|K| \leq |H| \leq \frac{k}{2p}$. So, $\E[|K|] \leq p\frac{k}{2p} = \frac{k}{2}$. Therefore on applying a Chernoff bound (\prop{chernoff}),
    \begin{align*}
        \Pr[|K| \geq k] \leq e^{-k/6}
    \end{align*}
    In addition, since every index $i \in [N]$ is put in $S$ independently with probability $p$, so $\E[|S|] \leq pN$. Again by a Chernoff bound (\prop{chernoff}),
    \begin{align*}
        \Pr[|S| \geq 2pN] \leq e^{-k/6}
    \end{align*}
    Finally, assuming $|S| \leq 2pN$, for every $y \in \{0,1\}^N$ with $\{i \in [N]: x_i \neq y_i\} \subseteq S \setminus K$, we have:
    \begin{align*}
        |\Tilde{f}(x) - \Tilde{f}(y)| &\leq \sum_{i:S\setminus K} c_{x,i} \\
        &\leq \tau |S| \\
        &\leq \frac{4 p^{7/4} F N}{k}
    \end{align*}
    Therefore by a union bound, with probability at least $1-2e^{-k/6}$ over choice of $S$, the statement holds.
\end{proof}

Similar to the remark we made for \thm{randrestbqp} above, we have for \lem{restlowdeg} that $|K| \geq 1$ with probability at most $\frac{k}{2}p^{1/4}$ by a union bound. In particular, if $k = \polylog(N)$ and $p$ is sufficiently small (say $1/N^{3/4}$), then with reasonably high probability over the choice of $S$ (say $1/N^{2/3}$), none of the indices $i$ such that $c_{x,i} > \tau$ are included in $S$. Therefore, the restricted function becomes a constant function in this case.

\begin{theorem}\label{thm:restlowdeg}
    Let $f:\{0,1\}^N \rightarrow \{0,1,\perp\}$ be a partial function which is approximated to error 1/3 by a real-valued function $\Tilde{f}$ of fractional block sensitivity $F$. Pick a set $S \subseteq [N]$ where each index $i \in [N]$ is put in $S$ independently with probability $p \leq (\frac{k}{48FN})^{4/7}$, where $k \in \mathbb{N}$. Then
    \begin{align*}
        \forall x, y \in \{0,1\}^N \Pr_S[g(y_{|S}) = f_{\rho}(y_{|S})] \geq 1-(2+\frac{k}{6})e^{-k/6}
    \end{align*}
    where $y_{|S}$ is the string $y$ restricted to indices in $S$, $\rho$ is a $p$-random restriction defined from $x$ and $S$ as follows:
    \begin{align*}
        \rho(i) = \begin{cases}
            * & \text{if $i \in S$} \\
            x_i & \text{otherwise}
        \end{cases}
    \end{align*}
    and $g$ is a width-$k^2$ DNF dependent on $\rho$ (defined explicitly in the proof). Note that if $f_{\rho}(y_{|S}) = \perp$, then $g$ is allowed to output 0 or 1 arbitrarily.
\end{theorem}
\begin{proof}
    Choose arbitrary $x, y \in \{0,1\}^N$, and choose $S$ as described in the theorem statement to define $\rho$. Suppose $f$ is approximated to error 1/3 by a real-valued function $\Tilde{f}$ of fractional block sensitivity $F$. Let $C$ be the set of all 1-certificates for $\Tilde{f}$ after applying the restriction $\rho$. Then the DNF $g$ is defined as follows (this is the same DNF that \cite{AIK22} show closeness to):
    \begin{align*}
        g(y) = \bigvee_{\substack{(K_{x'},x') \in C \\ |K_{x'}| \leq k^2}}~\bigwedge_{i \in K_{x'}} y_i = x'_i
    \end{align*}
    Note that if $f_{\rho}(y_{|S}) = 0$, then $g$ outputs 0 because it will not find a 1-certificate in $y_{|S}$. In addition, we do not care about the output of $g$ when $f_{\rho}(y_{|S}) = \perp$. So now we only worry about the case when $f_{\rho}(y_{|S}) = 1$. Define $h(x) = \{i \in [N]: c_{x,i} > \tau\}$, where $\tau = \frac{2p^{3/4}F}{k}$ and $c_{x,i}$ is the fractional weight of $\Tilde{f}$ on index $i$ for input $x$. We now think of sampling $S$ in stages (instead of all at once), and we start by sampling $S_0$ for indices in $h(x)$. Note that the set $S$ is still sampled by putting each bit $i$ in $S$ independently with probability $p$, we only adopt this viewpoint of stages to analyze the sampling. In particular, since each bit is included in $S$ independently, we can analyze the random restriction by looking at a subset of bits at a time and seeing whether they were included in $S$ or not, without affecting the inclusion/exclusion of other bits, which we can analyze in the next step. Recall, we start by sampling $S_0$ for indices in $h(x)$. Define $A_0 = \emptyset$, $A_1 = \{i \in S_0: x_i \neq y_i\}$ and $x^{A_i}$ as the string $x$ with bits in $A_i$ flipped. Then $S$ is sampled further in stage $j$ as follows: $S_j$ is obtained by sampling indices in $h(x^{A_j}) \setminus \cup_{l \leq j-1}h(x^{A_{l}})$, and $A_{j+1} = \{i \in \cup_{l \leq j} S_l: x_i \neq y_i\}$. Note that $\forall~j~|h(x^{A_j})| \leq \frac{F}{\tau}$, and therefore by a union bound, $|S_j| \geq 1$ with probability at most $p\frac{F}{\tau} = \frac{k}{2}p^{1/4} < \frac{1}{e}$ over choice of $S_j$ (we assume that $k$ is sufficiently small, say $k < N^{1/8}$). If $|S_j| = 0$, then we stop the stage-wise analysis and make a decision for all the input bits we haven't analyzed yet to get the full set $S$. Therefore, the probability that we reach stage $\frac{k}{6}$ of the sampling to sample $S_{\frac{k}{6}}$ is at most $e^{-k/6}$, since the sampling in each stage is performed independently of the previous stages. Note that the decision of whether or not to put any given bit in $S$ is independent of the other bits, by stages of the sampling we only change the order in which we analyze this decision. Further, in each stage of sampling, by \lem{restlowdeg}, $|S_j| > k$ with probability at most $e^{-k/6}$, thus the probability that any of $S_j$ has size more than $k$ for $j < k/6$ is at most $\frac{k}{6}e^{-k/6}$ by union bound. Therefore, with probability at least $1-(\frac{k}{6}+1)e^{-k/6}$, $A_{k/6} = A_{k/6-1}$ and $|\cup_{l \leq k/6-1} S_l| \leq \frac{k}{6}k$ (since $|S_j| \leq k$ for all $j \leq k/6-1$). Now we sample the remaining indices of $S$, and by \lem{restlowdeg}, $|S| > 2pN$ with probability at most $e^{-k/6}$. So now we assume that $A_{k/6} = A_{k/6-1}$, $|\cup_{l \leq k/6-1} S_l| \leq \frac{k}{6}k$ and $|S| \leq 2pN$, which happens with probability at least $1-(\frac{k}{6}+2)e^{-k/6}$. For convenience, we now set $j = k/6-1$. Define $y' \in \{0,1\}^N$ as follows:
        \begin{align*}
            y'_i = \begin{cases}
                y_i & \text{if $i \in S$} \\
                x_i & \text{otherwise}
            \end{cases}
        \end{align*}
    Therefore $y'_{|S} = y_{|S}$ and $f(y') = f_{\rho}(y_{|S})$. Let $K = h(x^{A_j}) \cap S$. Note that $K \subseteq \cup_{l \leq j} S_l$, because $S_j$ was sampled from bits in $h(x^{A_j}) \setminus \cup_{l \leq j-1}h(x^{A_{l}})$ and the bits of $h(x^{A_j})$ in $\cup_{l \leq j-1}h(x^{A_{l}})$ were sampled in $\cup_{l \leq j-1} S_l$. So $|K| \leq \frac{k}{6}k$ because we assume $|\cup_{l \leq j} S_l| \leq \frac{k}{6}k$. From \lem{restlowdeg}, $|\Tilde{f}(x^{A_j}) - \Tilde{f}(z)| \leq 1/12$ for all $z$ which differ from $x^{A_j}$ only on indices in $S \setminus K$. In particular, $y'$ differs from $x^{A_j}$ only on bits in $S \setminus K$, because $x^{A_j}$ agrees with $y$ on all bits in $\cup_{l \leq j} S_l$ and $K$ is a subset of these bits. Therefore $|\Tilde{f}(x^{A_j}) - \Tilde{f}(y')| \leq 1/12$. Therefore, if $f(y') = 1$, then $\Tilde{f}(z) > 1/2$ for all $z$ which differ from $x^{A_j}$ only on bits in $S \setminus K$. Thus, $(K, x^{A_j})$ is a 1-certificate of size at most $\frac{k^2}{6}$ for $\Tilde{f}$ when bits outside of $S$ are fixed to those of $x$. Finally, $y_{|S}$ is consistent with this certificate, thus $g(y_{|S}) = 1$. Since our assumption holds true with probability at least $1-(\frac{k}{6}+2)e^{-k/6}$, therefore, for every $x, y \in \{0,1\}^N$, the statement in the theorem holds true.    
\end{proof}

\noindent The argument above from \thm{restlowdeg} actually shows that a given 1-input of $f$ is consistent with a small 1-certificate (of size $k^2$) with high probability over the choice of unrestricted bits. A symmetric argument can be made for 0-certificates as well. Thus the theorem can be restated as follows:

\begin{corollary}\label{cor:smallcert}
    Let $f:\{0,1\}^N \rightarrow \{0,1,\perp\}$ be a partial function which is approximated to error 1/3 by a real-valued function $\Tilde{f}$ of fractional block sensitivity $F$. Pick a set $S \subseteq [N]$ where each index $i \in [N]$ is put in $S$ independently with probability $p \leq (\frac{k}{48FN})^{4/7}$, where $k \in \mathbb{N}$. Then
    \begin{align*}
        \forall x, y \in \{0,1\}^N \Pr_S[C_{y_{|S}}(f_{\rho}) \leq k^2] \geq 1-(2+\frac{k}{6})e^{-k/6}
    \end{align*}
    where $y_{|S}$ is the string $y$ restricted to indices in $S$, $\rho$ is a $p$-random restriction defined from $x$ and $S$ as follows:
    \begin{align*}
        \rho(i) = \begin{cases}
            * & \text{if $i \in S$} \\
            x_i & \text{otherwise}
        \end{cases}
    \end{align*}
\end{corollary}

\noindent In particular, since we now have small certificates (with high probability) for the restricted function, we can construct a decision tree of small-depth which is correct with high probability over the choice of random restriction and a uniformly random input.

\begin{corollary}\label{cor:dectreelowdeg}
    Let $f:\{0,1\}^N \rightarrow \{0,1,\perp\}$ be a partial function which is approximated to error 1/3 by a real-valued function $\Tilde{f}$ of fractional block sensitivity $F$. Pick a set $S \subseteq [N]$ where each index $i \in [N]$ is put in $S$ independently with probability $p \leq (\frac{k}{48FN})^{4/7}$, where $k \in \mathbb{N}$. Then
    \begin{align*}
        \forall x, y \in \{0,1\}^N \Pr_S[g(y_{|S}) = f_{\rho}(y_{|S})] \geq 1-(2+\frac{k}{6})e^{-k/6}
    \end{align*}
    where $y_{|S}$ is the string $y$ restricted to indices in $S$, $\rho$ is a $p$-random restriction defined from $x$ and $S$ as follows:
    \begin{align*}
        \rho(i) = \begin{cases}
            * & \text{if $i \in S$} \\
            x_i & \text{otherwise}
        \end{cases}
    \end{align*}
    and $g$ is a decision tree of depth-$k^4$ (described explicitly in the proof), dependent on $\rho$. Note that if $f_{\rho}(y_{|S}) = \perp$, then $g$ is allowed to output 0 or 1 arbitrarily.
\end{corollary}
\begin{proof}
    We start by noting that every 0-certificate of the restricted function must contradict every 1-certificate of the restricted function at atleast one index. Therefore the decision tree $g$ works in the standard manner: pick a 0-certificate of length at most $k^2$ and query all the corresponding input bits. If the algorithm finds the certificate then stop and output 0, otherwise the algorithm learns one bit of every 1-certificate. Keep doing this for $k^2$ rounds (every time picking a 0-certificate of length $k^2$ which is consistent with the bits queried so far), till either it finds a 0 or 1-certificate (in which case it stops and outputs the correct bit) or every 1-certificate of length at most $k^2$ has been contradicted (in which case it outputs 0). In case the algorithm runs out of 0-certificates of length at most $k^2$ to try in less than $k^2$ rounds, then output 1. Suppose $f_{\rho}(y_{|S}) = \perp$, then we do not care about the output of the algorithm. If $f_{\rho}(y_{|S}) = 0$, then either every 1-certificate of length $k^2$ gets contradicted, so the algorithm outputs 0 and it is correct, or the algorithm runs out of small 0-certificates to try before contradicting all small 1-certificates. In the latter case the algorithm is wrong, but this can only happen when $C_{y_{|S}}(f_{\rho}) > k^2$, which happens with probability at most $(2+\frac{k}{6})e^{-k/6}$ from \cor{smallcert}. So now we consider the case when $f_{\rho}(y_{|S}) = 1$. In this case, if after $k^2$ rounds the algorithm finds a 1-certificate or if it runs out of small 0-certificates to try then it will output 1, so it will only be wrong if every 1-certificate of length at most $k^2$ has been contradicted. But this happens only when $C_{y_{|S}}(f_{\rho}) > k^2$, which happens with probability at most $(2+\frac{k}{6})e^{-k/6}$ from \cor{smallcert}.
\end{proof}

As a corollary of \thm{restlowdeg}, we can also conclude that functions which have low quantum query complexity also become close to a decision tree of small depth, after applying a random restriction.

\begin{theorem}\label{thm:restbqpstronger}
    \thm{restlowdeg} also holds for partial functions $f:\{0,1\}^N \rightarrow \{0,1,\perp\}$ of quantum query complexity $T$, when we set $p \leq (\frac{k}{48\pi^2T^2N})^{4/7}$ to obtain a width-$k^2$ DNF.
\end{theorem}
\begin{proof}
    Let $Q$ be the quantum query algorithm computing $f$ by making $T$ queries to the input. Let $\Tilde{f}$ be the polynomial of degree $d \leq 2T$ corresponding to $Q$ (from \lem{querytodeg}). Let $F$ be the fractional block sensitivity of $\Tilde{f}$. Then we know from \thm{fbstodeg}
    \begin{align*}
        F &\leq \frac{\pi^2}{4}d^2 \leq \pi^2T^2 \\
        p &\leq (\frac{k}{48\pi^2T^2N})^{4/7} \leq (\frac{k}{48FN})^{4/7}
    \end{align*}
    Therefore using \thm{restlowdeg}, the claim follows.
\end{proof}
\begin{theorem}\label{thm:dectreebqp}
    \cor{dectreelowdeg} also holds for partial functions $f:\{0,1\}^N \rightarrow \{0,1,\perp\}$ of quantum query complexity $T$, when we set $p \leq (\frac{k}{48\pi^2T^2N})^{4/7}$ to obtain a depth-$k^4$ decision tree.
\end{theorem}
\begin{proof}
    Follows similarly to \thm{restbqpstronger} and \cor{dectreelowdeg}.
\end{proof}

\paragraph{Remarks.} We now comment on why our proof for the switching lemma from \thm{restbqpstronger} does not immediately extend to $\QMA$-query algorithms, unlike Theorem 65 of \cite{AIK22}. Note that we need to iteratively analyse the random restriction, where in each step we flip some bits of the input $x$ according to the string $y$. As such, a priori it is unknown at what input $y'$ we will eventually run the quantum algorithm. Therefore, it is unclear how to fix a quantum proof $\ket{\psi}$ (which depends on $y'$) to provide to the $\QMA$-query algorithm. However, the output and query magnitudes of the algorithm will depend on the choice of this quantum proof. In \cite{AIK22} on the other hand, they only consider one input $x$, which is used both to set the restricted bits and the input to the unrestricted bits, so they can fix a quantum proof $\ket{\psi_x}$ according to the input $x$.

\section{Random projections for quantum query algorithms}
We now consider the effect of random projections on quantum query algorithms, and we start by showing the following theorem on block random restrictions (where we restrict an entire block of bits at a time, instead of restricting individual bits) for quantum query algorithms. This proof is essentially the same as \thm{restlowdeg}.

\begin{theorem}\label{thm:projbqpstronger}
    Let $f:\{0,1\}^{N \times l} \rightarrow \{0,1,\perp\}$ be a partial function computable by a quantum query algorithm making $T$ queries to the input. Pick a set $S \subseteq [N]$ where each index $i \in [N]$ is put in $S$ independently with probability $p = (\frac{k}{192^2T^2N})^{4/7}$, where $k \in \mathbb{N}$. Then
    \begin{align*}
        \forall x, y \in \{0,1\}^{N \times l} \Pr_S[g(y_{|S}) = f_{\rho}(y_{|S})] \geq 1-(2+\frac{k}{6})e^{-k/6}
    \end{align*}
    where $y_{|S}$ is the string $y$ restricted to blocks in $S$, $\rho$ is a $p$-block-random restriction defined from $x$ and $S$ as follows:
    \begin{align*}
        \rho(i,j) = \begin{cases}
            * & \text{if $i \in S$} \\
            x_{ij} & \text{otherwise}
        \end{cases}
    \end{align*}
    and $g$ is a width-$lk^2$ DNF dependent on $\rho$ (defined explicitly in the proof). Note that if $f_{\rho}(y_{|S}) = \perp$, then $g$ is allowed to output 0 or 1 arbitrarily.
\end{theorem}
\begin{proof}
    Choose arbitrary $x, y \in \{0,1\}^{N \times l}$, and choose $S$ as described in the theorem statement to define $\rho$. We will think of input bits as being divided into $N$ blocks, with each block having $l$ bits, and thus index the bits using $(i,j)$ where $i \in [N], j \in [l]$. Let $f$ be computable by the quantum query algorithm $Q$ and let $C$ be the set of all 1-certificates for the behaviour of $Q$ after applying the restriction $\rho$ (we can let $\Tilde{f}$ be the polynomial corresponding to $Q$ and look at the certificates for $\Tilde{f}$). Then the DNF $g$ is defined as follows (this is the same DNF that \cite{AIK22} show closeness to):
    \begin{align*}
        g(y) = \bigvee_{\substack{(K_{x'},x') \in C \\ |K_{x'}| \leq lk^2}}~\bigwedge_{i \in K_{x'}} y_i = x'_i
    \end{align*}
    We now define block query magnitudes for the quantum algorithm $Q$ on an input $x$ as $q_i(x) = \sum_{j=1}^l m_{ij}(x)$ for all $i \in [N]$, where $m_{ij}(x)$ is the query magnitude of the quantum algorithm on the $j^{th}$ bit in the $i^{th}$ block. We now think of sampling $S$ in a stage-wise manner similar to the proof of \thm{restlowdeg}, using $q_i(x)$ instead of $c_{x,ij}$ (which were the fractional certificates for a real-valued approximation of $f$). Following the same argument, define $h(x) = \{i \in [N]: q_i(x) > \tau\}$, where $\tau = \frac{2p^{3/4}T}{k}$ and $q_i(x)$ is the block query magnitude of $Q$ on index $i$ for input $x$. We now think of sampling $S$ in stages, and we start by sampling $S_0$ for indices in $h(x)$. Note that the set $S$ is still sampled by putting each index $i$ in $S$ independently with probability $p$, we only adopt this viewpoint of stages to analyze the sampling. In particular, since each index is included in $S$ independently, we can analyze the random restriction by looking at a subset of indices at a time and seeing whether they were included in $S$ or not, without affecting the inclusion/exclusion of other indices, which we can analyze in the next step. Recall, we start by sampling $S_0$ for indices in $h(x)$. Define $A_0 = \emptyset$, $A_1 = \{(i, j) \in S_0 \times [l]: x_{ij} \neq y_{ij}\}$ and $x^{A_i}$ as the string $x$ with bits in $A_i$ flipped. Then $S$ is sampled further in stage $j$ as follows: $S_j$ is obtained by sampling indices in $h(x^{A_j}) \setminus \cup_{l \leq j-1}h(x^{A_{l}})$, and $A_{j+1} = \{(i, j) \in (\cup_{m \leq j} S_m) \times [l]: x_{ij} \neq y_{ij}\}$. Note that $\forall~j~|h(x^{A_j})| \leq \frac{T}{\tau}$, and therefore by a union bound, $|S_j| \geq 1$ with probability at most $p\frac{T}{\tau} = \frac{k}{2}p^{1/4} < \frac{1}{e}$ over choice of $S_j$ (we assume that $k$ is sufficiently small, say $k < N^{1/8}$). If $|S_j| = 0$, then we stop the stage-wise analysis and make a decision for all the indices we haven't analyzed yet to get the full set $S$. Therefore, the probability that we reach stage $\frac{k}{6}$ of the sampling to sample $S_{\frac{k}{6}}$ is at most $e^{-k/6}$, since the sampling in each stage is performed independently of the previous stages. Note that the decision of whether or not to put any given index in $S$ is independent of the other indices, by stages of the sampling we only change the order in which we analyze this decision. Further, in each stage of sampling, by \thm{randrestbqp}, $|S_j| > k$ with probability at most $e^{-k/6}$, thus the probability that any of $S_j$ has size more than $k$ for $j < k/6$ is at most $\frac{k}{6}e^{-k/6}$ by union bound. Therefore, with probability at least $1-(\frac{k}{6}+1)e^{-k/6}$, $A_{k/6} = A_{k/6-1}$ and $|\cup_{m \leq k/6-1} S_m| \leq \frac{k}{6}k$ (since $|S_j| \leq k$ for all $j \leq k/6-1$). Now we sample the remaining indices of $S$, and by \thm{randrestbqp}, $|S| > 2pN$ with probability at most $e^{-k/6}$. So now we assume that $A_{k/6} = A_{k/6-1}$, $|\cup_{m \leq k/6-1} S_m| \leq \frac{k}{6}k$ and $|S| \leq 2pN$, which happens with probability at least $1-(\frac{k}{6}+2)e^{-k/6}$. For convenience, we now set $j = k/6-1$. Define $y' \in \{0,1\}^{N \times l}$ as follows:
        \begin{align*}
            y'_{ij} = \begin{cases}
                y_{ij} & \text{if $i \in S$} \\
                x_{ij} & \text{otherwise}
            \end{cases}
        \end{align*}
    Therefore $y'_{|S} = y_{|S}$ and $f(y') = f_{\rho}(y_{|S})$. Let $K = (h(x^{A_j}) \cap S) \times [l]$. Note that $K \subseteq (\cup_{l \leq j} S_l) \times [l]$, because $S_j$ was sampled from bits in $h(x^{A_j}) \setminus \cup_{l \leq j-1}h(x^{A_{l}})$ and the bits of $h(x^{A_j})$ in $\cup_{l \leq j-1}h(x^{A_{l}})$ were sampled in $\cup_{l \leq j-1} S_l$. So $|K| \leq l\frac{k}{6}k$ because we assume $|\cup_{l \leq j} S_l| \leq \frac{k}{6}k$. From \thm{randrestbqp}, $|\Pr[Q(x^{A_j}) = 1] - \Pr[Q(z) = 1]| \leq 1/12$ for all $z$ which differ from $x^{A_j}$ only on indices in $(S \times [l]) \setminus K$. In particular, $y'$ differs from $x^{A_j}$ only on bits in $(S \times [l]) \setminus K$, because $x^{A_j}$ agrees with $y$ on all bits in $(\cup_{l \leq j} S_l) \times [l]$ and $K$ is a subset of these bits. Therefore $|\Pr[Q(x^{A_j}) = 1] - \Pr[Q(y') = 1]| \leq 1/12$. Therefore, if $f(y') = 1$, then $\Pr[Q(z)] > 1/2$ for all $z$ which differ from $x^{A_j}$ only on bits in $(S \times [l]) \setminus K$. Thus, $(K, x^{A_j})$ is a 1-certificate of size at most $l\frac{k^2}{6}$ for behaviour of $Q$ when bits outside of $S$ are fixed to those of $x$. Finally, $y_{|S}$ is consistent with this certificate, thus $g(y_{|S}) = 1$. Since our assumption holds true with probability at least $1-(\frac{k}{6}+2)e^{-k/6}$, therefore, for every $x, y \in \{0,1\}^{N \times l}$, the statement in the theorem holds true.
\end{proof}

We now state an average case version of \thm{projbqpstronger}, with respect to two (potentially distinct) block-product distributions from which the strings $x$ and $y$ are sampled. An analogous proof for random restrictions where both the distributions are uniform will recover \cor{UniformSwitching}. Note below that $\{*_{1/2},1_{1/2}\}^l\setminus\{1\}^l$ denotes the product distribution where each bit is set independently to $\ast$ or $1$ with probability 1/2, conditioned on not setting every bit to $1$.

\begin{lemma}\label{lem:projdisagrdnf}
    Let $D_1 = \otimes_{i=1}^N D_{1i}$ and $D_2 = \otimes_{i=1}^N D_{2i}$ be two block-product distributions on $\{0,1\}^{N \times l}$. Let $\rho$ be a $p$-block-random restriction with underlying distribution $D_1$, where each block is sampled from $\{*_{1/2},1_{1/2}\}^l\setminus\{1\}^l$ with probability $p$ and from the corresponding block of $D_1$ otherwise. Let $D$ be the distribution induced by $D_2$ on the indices set to $*$ by $\rho$. Let $f:\{0,1\}^{N \times l} \rightarrow \{0,1,\perp\}$ be a partial function computable by a quantum query algorithm making $T$ queries to the input. Set $p = (\frac{k}{192^2T^2N})^{4/7}$, then there exists a DNF $g$ of width-$lk^2$ such that
    \begin{align*}
        \E_{\rho}[\Pr_{z \sim D}[f_{\rho}(z) \neq g(z)]] \leq (\frac{k}{6}+2)e^{-k/6}
    \end{align*}
    where $g$ is allowed to answer 0 or 1 arbitrarily if $f_{\rho}(z) = \perp$.
\end{lemma}
\begin{proof}
    We think of sampling $\rho$ as follows: Sample $x \sim D_1$, $y \sim (\{*_{1/2},1_{1/2}\}^l\setminus\{1\}^l)^{\otimes N}$ (call this distribution $D_3$) and $S \subseteq [N]$ by putting each index $i \in [N]$ in $S$ with probability $p$, and define $\rho$ from $x, y$ and $S$
    \begin{align*}
        \rho(i,j) = \begin{cases}
            y_{ij} & \text{if $i \in S$} \\
            x_{ij} & \text{otherwise}
        \end{cases}
    \end{align*}
    Let $h$ be the DNF of width $lk^2$ defined in \thm{projbqpstronger} corresponding to $x$ and $S$ and $g$ be the DNF obtained after fixing bits of $h$ that are additionally fixed by $\rho$ due to $y$. Then we need to upper bound the following quantity:
    \begin{align*}
        \E_{\rho}[\Pr_{z \sim D}[f_{\rho}(z) \neq g(z)]] &=  \Pr_{\rho}[\Pr_{z \sim D}[f_{\rho}(z) \neq g(z)]] \\
        &= \Pr_{x \sim D_1, y \sim D_3, S, z \sim D}[f_{\rho}(z) \neq g(z)] \\
        &= \Pr_{x \sim D_1, y \sim D_3, z \sim D_2, S}[f_{\rho}(z_{|S,y}) \neq g(z_{|S,y})] \\
        &\leq (\frac{k}{6}+2)e^{-k/6}
    \end{align*}
    where $z_{|S,y}$ refers to the string obtained by restricting $z$ to the indices set to $*$ by $\rho$. Let $r$ be the string obtained by setting $r_{ij} = z_{ij}$ when $y_{ij} = *$ and 1 otherwise. Then the last line follows by applying \thm{projbqpstronger} on strings $x$ and $r$ (note that $z_{|S,y} = r_{|S,y}$).
\end{proof}

In particular, on applying $\mathsf{proj}_{\rhoin}$ to a quantum query algorithm, the resulting function is close to a DNF of width $k^2$.
\begin{corollary}\label{cor:projdisagrdnf}
    Let $D_1 = \otimes_{i=1}^N D_{1i}$ be a block-product distribution on $\{0,1\}^{N \times l}$ and $D_2 = \otimes_{i=1}^N D_{2i}$ be a product distribution on $\{0,1\}^N$. Let $\rho$ be a $p$-block-random restriction with underlying distribution $D_1$, where each block is sampled from $\{*_{1/2},1_{1/2}\}^l\setminus\{1\}^l$ with probability $p$ and from the corresponding block of $D_1$ otherwise. Let $D$ be the distribution induced by $D_2$ on the indices set to $*$ by $\mathsf{proj}_{\rho}$. Let $f:\{0,1\}^{N \times l} \rightarrow \{0,1,\perp\}$ be a partial function computable by a quantum query algorithm making $T$ queries to the input. Set $p = (\frac{k}{192^2T^2N})^{4/7}$, then there exists a DNF $g$ of width-$k^2$ such that
    \begin{align*}
        \E_{\rho}[\Pr_{z \sim D}[\mathsf{proj}_{\rho}(f)(z) \neq g(z)]] \leq (\frac{k}{6}+2)e^{-k/6}
    \end{align*}
    where $g$ is allowed to answer 0 or 1 arbitrarily if $\mathsf{proj}_{\rho}(f)(z) = \perp$.
\end{corollary}
\begin{proof}
    Consider the DNF $h$ of width-$lk^2$ obtained from \lem{projdisagrdnf}, and define $g$ to be the DNF of width-$k^2$ obtained as a projection of $h$, i.e., $g = \mathsf{proj}(h)$, so every variable in $h$ from a given block is mapped to the same variable in $g$. Then the claim follows from \lem{projdisagrdnf}.
\end{proof}

\subsection{Random projections for \texorpdfstring{$\AC^0$}{AC0} circuits}
In our proof we rely heavily on the notation as well as results from \cite{HRST17}, therefore in this section we mention all the notation and results we use. Note that we reference the \href{https://arxiv.org/abs/1504.03398}{arXiv} version of the paper in the rest of the discussion.

\begin{definition}[$\SIP_d$ functions \cite{HRST17}]
    The $\SIP_d$ function is a depth $d$ read-once monotone formula with $n$ variables, with alternating layers of AND and OR gates. The bottom layer (adjacent to input variables) consists of AND gates, so the root is an OR gate if $d$ is even and AND gate otherwise. All the gates at a particular depth have the same fan-in, which is denoted by $w_i$ for gates of depth $i$. So the bottom fan-in is $w_{d-1}$ and the top fan-in is $w_0$. The parameters $t_k$ correspond to the bias of the distribution for the random projection for depth $k$, see \defn{subsequentproj}.
\begin{align*}
    w_{d-1} &:= m \\
    w &:= \lfloor m2^m/\log e \rfloor \\ \displaybreak
    q &:= \sqrt{p} = 2^{-m/2} = \Theta \left(\sqrt{\frac{\log w}{w}}\right) \\
    w_i &:= w \text{ for } 1\leq k \leq d-2 \\
    w_0 &:= \min_{i \in \mathbb{N}}\{(1-t_1)^{qi} \leq \frac{1}{2}\} = 2^m\ln(2)\cdot (1 \pm o_m(1)) \\
    \lambda &:= \frac{(\log w)^{3/2}}{w^{5/4}} \\
    t_{d-1} &:= \frac{p-\lambda}{q} = \Theta \left(\sqrt{\frac{\log w}{w}}\right) \\
    t_{k-1} &:= \frac{(1-t_k)^{qw}-\lambda}{q} = \Theta \left(\sqrt{\frac{\log w}{w}}\right) \text{ for } 2\leq k \leq d-1 \\
    n &= \frac{1\pm o_m(1)}{\log e} \cdot \left(\frac{m2^m}{\log e}\right)^{d-1}
\end{align*}
\end{definition}

Next we describe the addressing scheme used for the gates and input variables of $\SIP_d$ (which we will also use for $\SIP'_d$). Let $A_0 = \{\mathsf{output}\}$, and for $1 \leq k \leq d$, let $A_{k} = A_{k-1} \times [w_{k-1}]$. An element of $A_k$ specifies the address of a gate at depth $k$; so $A_{d} = \{{\mathsf{output}}\} \times [w_0] \times \cdots \times [w_{d-1}]$ is the set of addresses of the input variables. For some string $\tau \in \{0,1, \ast\}^{A_k} = \{0,1, \ast\}^{A_{k-1} \times [w_{k-1}]}$, use $\tau_a$ for $a \in A_{k-1}$ to denote the $a^{th}$ block of $\tau$ of length $w_{k-1}$. 

The symbol $\{0_p, 1_{1-p}\}^n$ denotes the random bit string of length $n$ where each bit is sampled independently, and is $0$ with probability $p$ and $1$ with probability $1-p$. The symbol $\{0_p, 1_{1-p}\}^n \setminus \{x\}$ denotes the product distribution conditioned on not outputting the string $x$. For our setting, $x$ will be either $0^n$ or $1^n$. The notation $x = a \pm b$ is shorthand for $x \in [a-b, a+b]$.

We modify the first random projection from that of \cite{HRST17}, to replace the bottom layer of quantum circuits in $\QAC^0$ circuits by DNFs. Our subsequent random projections are the same as that of \cite{HRST17}. Since the random projections are defined adaptively based on the outcome of the previous random projection, they define the notion of lift of a random restriction, which tells us the value taken by each gate in the bottom layer of the circuit to which the corresponding random projection has been applied, and this is used to decide how to sample the next random restriction.

\begin{definition}[Lift of a restriction, Definition 7 from \cite{HRST17}]
    Let $2\le k \le d$ and $\tau \in \{0,1,\ast\}^{A_{k}}$. Assume that the gates of $\SIP_d$ (or $\SIP'_d$) at depth $k-1$ are $\wedge$ gates (otherwise the roles of 0 and 1 below are reversed). The lift of $\tau$ is the string $\hat{\tau}\in \{0,1,\ast\}^{A_{k-1}}$ defined as follows:  for each $a \in A_{k-1}$, the coordinate $\hat{\tau}_a$ of $\hat{\tau}$ is
\[
\hat{\tau}_a =
\begin{cases}
0 & \text{if~} \tau_{a,i}=0 \text{~for any~}i \in [w_{k-1}]\\
1 & \text{if~} \tau_a = \{1\}^{w_{k-1}}\\
\ast & \text{if~} \tau_a \in \{\ast,1\}^{w_{k-1}} \setminus \{1^{w_{k-1}}\}.
\end{cases}
\]
\end{definition}

To show that the $\SIP_d$ function retains structure after applying a random projection, they define the notion of a typical restriction. Roughly, on applying a random projection corresponding to a typical restriction to $\SIP_d$, the depth of the formula is reduced by 1, and the bottom layer of $\wedge$ gates of $\SIP_d$ takes on values in $\{0,1,\ast\}$ such that the bottom fan-in of the projected formula is approximately $qw$. In addition, the fan-in of the layers above remains roughly the same. More generally, after applying a series of random projections to $\SIP_d$, all of which correspond to typical restrictions, applying the next random projection corresponding to a typical restriction ensures that the resulting formula has depth reduced by 1, bottom fan-in approximately $qw$ and fan-ins of all other layers remain roughly the same. \cite{HRST17} show that each of their random restrictions is typical with high probability, assuming that the previous restriction is typical. We will show that the first random restriction which we redefine, is also typical with high probability, and therefore all the subsequent random restrictions remain typical using the results of \cite{HRST17}. Note that in the discussion above and the following definition, we talk about projections corresponding to restrictions which are typical, \cite{HRST17} however, talk about projections corresponding to restrictions for which the lift is typical. So though \defn{typical} is identical to theirs, it is stated slightly differently here.

\begin{definition}[Typical random restriction, Definition 14 of \cite{HRST17}]\label{def:typical}
    Let $\tau \in \{0,1,\ast\}^{A_k}$ where $3 \leq k \leq d$. The restriction $\tau$ is typical if
    \begin{enumerate}
        \item (Bottom fan-in after projection $\simeq qw$) For all $a \in A_{k-2}$
        \begin{align*}
            |\hat{\tau}_a^{-1}(\ast)| = qw \pm w^{\beta(k-1,d)}~~~where~~~\beta(k,d) := \frac{1}{3} + \frac{d-k-1}{12d}
        \end{align*}
        \item (Preserves rest of the structure) For all $a \in A_{k-3}$
        \begin{align*}
             w_{k-3} - w^{4/5} \leq |(\hat{\hat{\tau}}_a)^{-1}(\ast)| \leq w_{k-3}
        \end{align*}
    \end{enumerate}
\end{definition}
These conditions also imply that all the gates in $A_{k-3}$ remain undetermined. This is because suppose $\tau$ is applied to a layer of $\wedge$ gates. Then by condition 1, the only values that $\vee$ gates in $A_{k-2}$ can get are $\ast$ or $1$ (so $\wedge$ gates in $A_{k-3}$ have inputs from $\ast$ and $1$). By condition 2, $\wedge$ gates in $A_{k-3}$ have at least one $\ast$ input, and since none of their inputs is $0$, they remain undetermined.

Finally, \cite{HRST17} show that if an OR function (or its restriction) is close to unbiased under some input distribution where each bit of the input is independent and identically distributed, then it has a small correlation with CNFs of small width under this distribution.

\begin{lemma}[Proposition 11.1 of \cite{HRST17}]\label{lem:corcnfor}
    Let $C: \{0,1\}^{n} \rightarrow \{0,1\}$ be a CNF of width-$r$ and $\tau \in \{0,1,\ast\}^{n}$. Let $\mathsf{OR}$ be the $\vee$ function on n bits and $\mathbf{Y} \leftarrow \{0_{1-p}, 1_{p}\}^{n}$ for $p \in [0,1]$, then
    \begin{align*}
        \Pr_{\mathbf{Y}}[(\mathsf{OR}_\tau (\mathbf{Y}) \neq C(\mathbf{Y})] \geq \mathsf{bias}(\mathsf{OR}_\tau, \mathbf{Y}) - rp
    \end{align*}
\end{lemma}

In order to get an average case depth-hierarchy theorem for $\AC^0$ circuits, \cite{HRST17} show three properties for their sequence of random projections:
\begin{enumerate}
    \item The sequence of random projections completes to the uniform distribution.
    \item The target function $\SIP_d$ remains "hard" to compute after applying the sequence of random projections.
    \item The approximating circuit $C$ trying to compute $\SIP_d$ "simplifies" greatly after applying the sequence of random projections.
\end{enumerate}
We first re-establish Property 1 after modifying the initial random projection to be applied. Then we establish Property 2 for our modified $\SIP'_d$ function, by showing that the effect of the modified initial random projection on $\SIP'_d$ is roughly the same as the effect of the initial random projection of \cite{HRST17} on $\SIP_d$. Since our subsequent random projections are the same as theirs, we can conclude that $\SIP'_d$ retains structure under this sequence of random projections. Finally, we use our random projection result for quantum query algorithm to show that the bottom layer of a $\QAC^0$ circuit simplifies after applying the initial random projection. Then we can use the results of \cite{HRST17} to conclude that the resulting $\AC^0$ circuit also simplifies after applying the subsequent random projections.

The other results we use from \cite{HRST17} have been stated and used in the following subsection, to maintain continuity in the proof.

\subsection{Random projections for \texorpdfstring{$\QAC^0$}{QAC0} circuits}

We start by showing an analogue of \cite{furstParityCircuitsPolynomialtime1984} for $\QCPH$, to use a circuit lower bound for $\QAC^0$ to get oracle separation result for $\QCPH$. Note that \cite{AIK22} show a similar analogue for the $\QMA$-hierarchy.

\begin{proposition}[Analogue of \cite{furstParityCircuitsPolynomialtime1984}]\label{prop:qcphcirc}
    Let $L \subseteq \{0\}^{*}$ be some language decided by a $\QCSigma_i$ verifier $V$ with oracle access to $O$, and which has size at most $p(n)$ for inputs of length $n$. Then for every $n \in \bN$, there is a circuit $C$ of size at most $2^{\poly (n)}$ and depth $i+1$, where each layer upto depth $i$ is a layer of AND or OR gates, and depth $i+1$ contains $p(n)$ sized quantum circuits with query complexity at most $p(n)$, such that $\forall~x \in \{0,1\}^n$
    \begin{align*}
        V^O(x) = C(O_{[\leq p(n)]})
    \end{align*}
    where $O_{[\leq p(n)]}$ denotes the concatenation of bits of $O$ on all strings of length at most $p(n)$.
\end{proposition}
\begin{proof}
    This proof is the same as that of \cite{furstParityCircuitsPolynomialtime1984}. The $j^{th}$ layer of the circuit $C$ ($j \leq i$) contains OR gates if $Q_j$ is an existential quantifier and AND gates otherwise. If the verifier $V$ takes $i$ proofs each of length $q(n)$, then each gate upto layer $i$ has fan-in $2^{q(n)}$, with one wire each for each possible choice of proof. For layer $i+1$, each gate is the quantum verifier $V$ with the input $x$ and the proofs corresponding to the root-to-gate path in $C$ hardcoded in it. This gate then makes quantum queries to the oracle $O$, and since this circuit has size at most $p(n)$, it can not make queries of length greater than $p(n)$ to $O$. Thus we can consider the circuit $C$ as taking inputs all the bits of $O_{[p(n)]}$ and outputting the decision of the $\QCSigma_i$ verifier. Since $q(n)$ is a polynomial, and this circuit has constant number of layers, the size of the circuit is at most $2^{\poly (n)}$.
\end{proof}

We now describe the modification of $\SIP_d$ function (used in \cite{HRST17}) that we show the $\QAC^0$ lower bound for. The modified function which we use will be called $\SIP'_d$ and is exactly the same as $\SIP_d$ defined previously, except for the bottom two fan-ins. The number of variables for $\SIP'_d$ will be denoted $N$.
\begin{align*}
    w &:= \lfloor m2^m/\log e \rfloor \\
    w_{d-2} &:= qwN^{5/7} \\
    q' &:= 1/N^{5/7} \\
    x &:= \frac{1}{w_{d-2}w^{1/4}} \\
    p_1 &:= x + q'(\frac{p-\lambda}{q}) \\ 
    w_{d-1} &:= - \log_2(p_1) = \polylog(N) \\
    N &= \frac{q}{q'}.\frac{n}{m}\log_2(\frac{1}{p_1})
\end{align*}

Note that $N$ and $n$ are polynomially related. We also redefine the first random projection that is applied to $\SIP'_d$, in order to ensure that $\SIP'_d$ retains structure, while the quantum query algorithms become simple enough to be replaced by DNFs.
\begin{definition}[Restriction for initial random projection]\label{def:initdist}
    The underlying random restriction $\rhoin$ on $\{0,1,*\}^{A_{d-1} \times [w_{d-1}]}$ is defined as follows: independently for each $a \in A_{d-1}$
\begin{align*}
    \rho(a) \leftarrow \begin{cases}
            \{1\}^{w_{d-1}} & \text{with probability $x$} \\
            \{*_{1/2},1_{1/2}\}^{w_{d-1}}\setminus\{1^{w_{d-1}}\} & \text{with probability $q'$} \\
            \{0_{1/2},1_{1/2}\}^{w_{d-1}}\setminus\{1^{w_{d-1}}\} & \text{with probability $1-x-q'$}
        \end{cases}
\end{align*}
We will call this distribution on random restrictions as $\mathcal{R}_{\mathsf{init}}$.
\end{definition}

We first show that the projection defined using $\rhoin \leftarrow \mathcal{R}_{\mathsf{init}}$ on composition with the subsequent random projections of \cite{HRST17} completes to the uniform distribution on $\{0,1\}^N$.

\begin{lemma}[Analogue of Lemma 8.2 of \cite{HRST17}]\label{lem:uniform}
    Let $\rho \leftarrow \mathcal{R}_{\mathsf{init}}$ and ${\mathbf{Y}} \leftarrow \{0_{1-t_{d-1}}, 1_{t_{d-1}}\}^{A_{d-1}}$, and let the string $\mathbf{X} \in \{0,1\}^N$ be defined as follows:
    \begin{align*}
        \mathbf{X}_{a,i} = \begin{cases}
            \mathbf{Y}_a & \text{if $\rho_{a,i} = *$} \\
            \rho_{a,i} & \text{otherwise}
        \end{cases}
    \end{align*}
    Then each bit of $\mathbf{X}$ is independent and distributed uniformly at random.
\end{lemma}
\begin{proof}
    The proof is similar to \cite{HRST17}. We show that $\mathbf{X}_a$ is distributed according to $\{0_{1/2},1_{1/2}\}^{w_{d-1}}$ for a fixed $a \in A_{d-1}$. Note that $\mathbf{X}_a = 1^{w_{d-1}}$ only in cases 1 and 2 of $\mathcal{R}_{\mathsf{init}}$.
    \begin{align*}
        \Pr[\mathbf{X}_a = 1^{w_{d-1}}] = x + q't_{d-1} = p_1 = 2^{\log(p_1)}
    \end{align*}
    For any other string $\mathsf{Z} \in \{0,1\}^{w_{d-1}}\setminus \{1\}^{w_{d-1}}$, $\mathbf{X}_a = 1^{w_{d-1}}$ only in cases 2 and 3 of $\mathcal{R}_{\mathsf{init}}$
    \begin{align*}
        \Pr[\mathbf{X}_a = \mathsf{Z}] &= (1-x-q')\frac{2^{\log(p_1)}}{1-2^{\log(p_1)}} + q'(1-t_{d-1})\frac{2^{\log(p_1)}}{1-2^{\log(p_1)}}\\
        &= \frac{2^{\log(p_1)}}{1-2^{\log(p_1)}} ((1-x-q') + q'(1-t_{d-1})) \\
        &= 2^{\log(p_1)}
    \end{align*}
    Now since $\rho$ and $\mathbf{Y}$ are sampled independently for each $a \in A_{d-1}$, $\mathbf{X}_a$ are independent for each $a \in A_{d-1}$, so $\mathbf{X}$ is distributed according to the uniform distribution $\{0_{1/2},1_{1/2}\}^N$.
\end{proof}

For completeness, we now describe the subsequent $(d-2)$ random projections applied to the $\SIP'_d$ function, which are the same as those of \cite{HRST17}.

\begin{definition}[Subsequent random projections, Definition 9 of \cite{HRST17}]\label{def:subsequentproj}
    Let $2 \leq k \leq d-1$. We assume that the gates of $\SIP'_d$ at depth $k-1$ are $\wedge$ gates (otherwise the roles of 0 and 1 below are reversed). Given a string $\tau \in \{0,1,\ast\}^{A_k}$, let $S_a = \{i \in [w_{k-1}]: \tau_{a,i} = \ast\}$ for $a \in A_{k-1}$. The distribution $\mathcal{R}(\tau)$ on random restrictions is defined as follows:
    \begin{itemize}
        \item For all $(a, i) \in A_{k-1} \times [w_{k-1}]$, if $\tau_{a,i} \neq \ast$, then $\rho(a,i) = \tau_{a,i}$. In particular, if $\hat{\tau}(a) = 1$ then $\rho(a,i) = \tau_{a,i} = 1$ for all $i \in [w_{k-1}]$.
        \item If $\hat{\tau}(a) = 0$ or if $|S_a| \neq qw \pm w^{\beta(k,d)}$ where $\beta(k,d)$ is the same as \defn{typical}, then each bit of $\rho$ in $S_a$ is set independently to $0$ with probability $t_k$ and $1$ with probability $1-t_k$.
        \item Otherwise,
        \begin{align*}
    \rho(S_a) \leftarrow \begin{cases}
            \{1\}^{S_a} & \text{with probability $\lambda$} \\
            \{*_{t_k},1_{1-t_k}\}^{S_a}\setminus\{1^{S_a}\} & \text{with probability $q_a$} \\
            \{0_{t_k},1_{1-t_k}\}^{S_a}\setminus\{1^{S_a}\} & \text{with probability $1-\lambda-q_a$}
        \end{cases}
\end{align*}
    where
    \begin{align*}
        q_a := \frac{(1-t_k)^{|S_a|}-\lambda}{t_{k-1}}
    \end{align*}
    \end{itemize}
\end{definition}

\cite{HRST17} show (in the proof of Proposition 8.4) that the subsequent ($d$-2) random projections after $\rhoin$ (along with a suitably chosen distribution for the unrestricted variables) complete to the distribution $\{0_{1-t_{d-1}}, 1_{t_{d-1}}\}^{(\widehat{\rhoin})^{-1}(*)}$, and therefore we can conclude that the overall random projection completes to the uniform distribution, using \lem{uniform}.

\begin{lemma}[Proposition 8.4 of \cite{HRST17}]\label{lem:uniformrst}
    Let $\tau \in \{0,1,\ast\}^{A_d}$. Sample $\rho^{(k)} \leftarrow \mathcal{R}(\widehat{\rho^{(k+1)}})$ for $2 \leq k \leq d-1$ where we denote $\rho^{(d)} = \tau$. Let $T_k = \{i \in [A_k]: \widehat{\rho^{(k+1)}}(i) = \ast\}$ for $1 \leq k \leq d-1$ and sample $\mathbf{Y}^{(1)} \leftarrow \{0_{1-t_1}, 1_{t_1}\}^{T_1}$ (assuming $d$ is even, otherwise 0 and 1 in $\mathbf{Y}^{(1)}$ are flipped). Define for $2 \leq k \leq d-1$, the string $\mathbf{Y}^{(k)} \in \{0,1\}^{T_k}$:
    \begin{align*}
        \mathbf{Y}^{(k)}_{a,i} = \begin{cases}
            \mathbf{Y}^{(k-1)}_a & \text{if $\rho^{(k)}(a,i) = \ast$} \\
            \rho^{(k)}(a,i) & \text{otherwise}
        \end{cases}~~~~~~~~~\forall (a,i) \in T_k
    \end{align*}
    Then $\mathbf{Y}^{(d-1)}$ is distributed according to $\{0_{1-t_{d-1}},1_{t_{d-1}}\}^{T_{d-1}}$.
\end{lemma}

For ease of writing, \cite{HRST17} give the following definition, as notation for the resulting function after applying all the random projections.
\begin{definition}[Definition 10 from \cite{HRST17}]\label{def:projnotation}
    Given a function $f:\{0,1\}^N \rightarrow \{0,1\}$, sample a sequence of $d$-1 random restrictions as follows:
    \begin{align*}
        \rho^{(d)} &\leftarrow \mathcal{R}_{\mathsf{init}} \\
        \rho^{(k)} &\leftarrow \mathcal{R}(\widehat{\rho^{(k+1)}}) &\forall 2 \leq k \leq d-1
    \end{align*}
    Then $\mathbf{\Psi}$ denotes the composition of random projections of the restrictions sampled above:
    \begin{align*}
        \mathbf{\Psi}(f) \equiv \mathsf{proj}_{\rho^{(2)}}\mathsf{proj}_{\rho^{(3)}}\cdots\mathsf{proj}_{\rho^{(d)}}f
    \end{align*}
\end{definition}

\begin{corollary}[Proposition 8.1 of \cite{HRST17}]\label{cor:uniformdist}
    Consider $C:\{0,1\}^N \rightarrow \{0,1\}$ be a circuit computing $\SIP'_d$ defined on $N$ variables. Let $\mathbf{X} \leftarrow \{0_{1/2},1_{1/2}\}^N$ and $\mathbf{Y} \leftarrow \{0_{1-t_1},1_{t_1}\}^{w_0}$ (we assume that $d$ is even). If $\mathbf{\Psi}$ is sampled according to \defn{projnotation},
    \begin{align*}
        \Pr_{\mathbf{X}}[\SIP'_d(\mathbf{X}) \neq C(\mathbf{X})] = \Pr_{\mathbf{\Psi},\mathbf{Y}}[(\mathbf{\Psi}(\SIP'_d))(\mathbf{Y}) \neq (\mathbf{\Psi}(C))(\mathbf{Y})]
    \end{align*}
\end{corollary}
\begin{proof}
    Follows from \lem{uniform}
    and \lem{uniformrst}.
\end{proof}

We now show that $\rhoin$ as sampled above, results in a typical $\rho$ with high probability. We establish the two conditions for typicality in the next two lemmas. We use $\mathbf{\tau}$ to denote $\hat{\rho} \in \{0,1,*\}^{A_{d-1}}$.

\begin{lemma}[Analogue of Lemma 10.3 of \cite{HRST17}]\label{lem:typical1}
    Fix $\alpha \in A_{d-2}$. Then
    \begin{align*}
        \Pr[|\mathbf{\tau}_\alpha^{-1}(*)| = qw \pm w^{1/3}] \geq 1-e^{-\Tilde{\Omega}(w^{1/6})}
    \end{align*}
\end{lemma}
\begin{proof}
    $\mathbf{\tau}(a)$ for $a \in A_{d-1}$ tells the value taken by the gate at address $a$ after applying the initial random projection, and by definition of $\rhoin$, this value is $*$ with probability $q'$, 1 with probability $x$ and 0 otherwise. So for the fixed $\alpha$ in the statement
    \begin{align*}
        \E[|\mathbf{\tau}_\alpha^{-1}(*)|] &= q'.w_{d-2} = qw = \mu \\
        \mu &= \Theta((w\log w)^{1/2}) \\
        \gamma \mu &= w^{1/3} \\
        \gamma &= \Theta(w^{-1/6}(\log w)^{-1/2})
    \end{align*}
    Therefore on applying a Chernoff bound (\prop{chernoff}),
    \begin{align*}
        \Pr[||\mathbf{\tau}_\alpha^{-1}(*)| - qw| > w^{1/3}] &\leq \exp(-\Omega(\gamma^2\mu)) = \exp(-\Tilde{\Omega}(w^{1/6}))\qedhere
    \end{align*}
\end{proof}

\begin{lemma}[Analogue of Lemma 10.4 of \cite{HRST17}]\label{lem:typical2}
    Fix $\alpha \in A_{d-3}$. Then
    \begin{align*}
        \Pr[w - w^{4/5} \leq |\hat{\mathbf{\tau}}_\alpha^{-1}(*)| \leq w] \geq 1-e^{-\Omega(w^{4/5})}
    \end{align*}
\end{lemma}
\begin{proof}
    We look at the OR gate at address $\alpha, i$ for $i \in [w_{d-3}]$. This gate is not determined to value 1 with probability  $(1-x)^{w_{d-2}} \geq 1-xw_{d-2} \geq 1-w^{-1/4}$. This gate is determined to 0 with probability $(1-x-q')^{w_{d-2}} \leq (1-q')^{w_{d-2}} \leq e^{-qw} \leq w^{-1/4}$. So independently for each $i \in [w_{d-3}]$,
    \begin{align*}
        \Pr[\hat{\mathbf{\tau}}_{\alpha, i} \in \{0,1\}] \leq O(w^{-1/4})
    \end{align*}
    Therefore,
    \begin{align*}
        \E[|\hat{\mathbf{\tau}}_\alpha^{-1}(\{0,1\})|] &\leq w_{d-3}O(w^{-1/4}) = O(w^{3/4}) \\
        \mu &= O(w^{3/4}) \\
        \gamma \mu &= \Omega(w^{4/5}) \\
        \gamma &= \Omega(w^{1/20})
    \end{align*}
    Therefore on applying a Chernoff bound (\prop{chernoff}),
    \begin{align*}
        \Pr[w - w^{4/5} \leq |\hat{\mathbf{\tau}}_\alpha^{-1}(*)| \leq w] &\geq 1-\exp(-\Omega(\gamma\mu)) = 1-\exp(-\Omega(w^{4/5}))\qedhere
    \end{align*}
\end{proof}

\begin{lemma}[Lemma 10.6 and Lemma 10.8 of \cite{HRST17}]\label{lem:typicalrst}
    For $2 \leq k \leq d-1$, let $\tau \in \{0,1,\ast\}^{A_{k+1}}$ be typical. Then $\rho \leftarrow \mathcal{R}(\hat{\tau})$ is typical with probability at least $1-e^{-\Tilde{\Omega}(w^{1/6})}$.
\end{lemma}

On application of the overall sequence of random projections $\mathbf{\Psi}$, the function $\SIP'_d$ becomes an OR of fan-in at most $w_0$ if $d$ is even, and an AND of fan-in at most $w_0$ otherwise. We will now assume that $\SIP'_d$ becomes an OR of fan-in at most $w_0$, the case for AND is analogous. Then we sample the first ($d$-2) restrictions and assume they are all typical, which happens with probability at least $1-de^{-\Tilde{\Omega}(w^{1/6})}$ using \lem{typical1}, \lem{typical2} and \lem{typicalrst} and a union bound. We know from \defn{typical} that the top OR gate of the function after applying the first ($d$-2) projections is undetermined and has fan-in at least $w_0 - w^{4/5}$. \cite{HRST17} then show that in expectation over the final random projection, the bias of the projected function is close to 1/2 under the distribution $\{0_{1-t_1}, 1_{t_1}\}^{w_0}$.

\begin{lemma}[Proposition 10.13 of \cite{HRST17}]\label{lem:unbiasedor}
    Let $\mathbf{Y} \leftarrow \{0_{1-t_1}, 1_{t_1}\}^{w_0}$ and $\mathbf{\Psi}(\SIP'_d)$ be the random projection of $\SIP'_d$ when $\mathbf{\Psi}$ is sampled according to \defn{projnotation}. Then
    \begin{align*}
        \E_{\mathbf{\Psi}}[\mathsf{bias}(\mathbf{\Psi}(\SIP'_d),\mathbf{Y})] \geq \frac{1}{2} - \Tilde{O}(w^{-1/12})
    \end{align*}
\end{lemma}

Now we state and use the projection switching lemma of \cite{HRST17} (statement taken from their paper), to obtain our final circuit lower bound.

\begin{theorem}[Proposition 9.2 of \cite{HRST17}]\label{thm:randomprojcnf}
Let $2 \leq k \leq d-1$ and $F:\{0,1\}^{A_k} \rightarrow \{0,1\}$ be a DNF/CNF of width $r$. Then for all $\tau \in \{0,1,*\}^{A_k}$ and $s \geq 1$,
    \begin{align*}
        \Pr_{\rho \leftarrow \mathcal{R}(\tau)}[\mathsf{DT}(\mathsf{proj}_{\rho}F) > s] = (O(re^{rt_k/(1-t_k)}w^{-1/4}))^s
    \end{align*}
\end{theorem}

We want to use \thm{randomprojcnf} ($d$-2) times on the circuit obtained after applying the first random projection and replacing the layer of quantum circuits in a $\QAC^0_{d-3}$ circuit with a small width DNF/CNF as per \cor{projdisagrdnf}, to conclude that this circuit becomes a decision tree with high probability.



\begin{definition}
    Given a function $f:\{0,1\}^{A_{d-1}} \rightarrow \{0,1\}$, and $\tau \in \{0,1,\ast\}^{A_d}$, sample a sequence of $d$-2 random restrictions $\rho^{(k)}$ for $2 \leq k \leq d-1$ as follows:
    \begin{align*}
        \rho^{(d)} &:= \tau \\
        \rho^{(k)} &\leftarrow \mathcal{R}(\widehat{\rho^{(k+1)}}) &\forall 2 \leq k \leq d-1
    \end{align*}
    Then $\mathbf{\Psi}_\tau$ denotes the composition of random projections, of the restrictions sampled above:
    \begin{align*}
        \mathbf{\Psi}_{\tau}(f) \equiv \mathsf{proj}_{\rho^{(2)}}\mathsf{proj}_{\rho^{(3)}}\cdots\mathsf{proj}_{\rho^{(d-1)}}f
    \end{align*}
\end{definition}

\begin{corollary}[Analogue of Proposition 9.13 of \cite{HRST17}]\label{cor:circuittocnf}
    For any constant $d \geq 2$, let $C: \{0,1\}^{A_{d-1}} \rightarrow \{0,1\}$ be a depth-$d$ circuit which has $\wedge$ $($or $\vee)$ as output gate, with bottom fan-in $w^{1/5}$ and size $S \leq 2^{w^{1/5}}$. Then for any $\tau \in \{0,1,*\}^{A_{d}}$
    \begin{align*}
        \Pr_{\mathbf{\Psi}_\tau}[\mathbf{\Psi_{\tau}}(C) \text{ is a CNF (or DNF) of width at most } w^{1/5}] \geq 1-e^{-\Omega(w^{1/5})}
    \end{align*}
\end{corollary}
\begin{proof}
    Apply \thm{randomprojcnf} ($d$-2) times on $C$, with $r = s = w^{1/5}$, along with a union bound on the number of gates in the bottom layer (at most $S$) for each application of \thm{randomprojcnf}.
\end{proof}

\begin{theorem}\label{thm:randomlowerbound}
    For any even constant $d \geq 4$, let $\SIP'_{d}$ be defined on $N$ variables. Let $C:\{0,1\}^N \rightarrow \{0,1\}$ be any $\QAC^0_{d-2}$ circuit of size $S = \mathsf{quasipoly}(N) < 2^{N^{\frac{1}{6(d-1)}}}$, which has $\wedge$ as its output gate. Then for a uniformly random input $\mathbf{X}$, we have
    \begin{align*}
        \Pr_{\mathbf{X}}[\SIP'_{d}(\mathbf{X}) \neq C(\mathbf{X})] \geq \frac{1}{2} - \frac{1}{N^{\Omega(1/d)}}
    \end{align*}
\end{theorem}
\begin{proof}
    Set $k = \lceil 6(\ln \frac{2S}{3} + \frac{1}{d} \ln N)\rceil \leq \polylog(N)$. Let $D_1 = \otimes_{i=1}^{A_{d-1}} D_{1i}$ where for each $i \in A_{d-1}$, $D_{1i}$ samples $\{1\}^{w_{d-1}}$ with probability $\frac{x}{1-q'}$ and $\{0_{1/2},1_{1/2}\}^{w_{d-1}}\setminus\{1^{w_{d-1}}\}$ with probability $\frac{1-x-q'}{1-q'}$. Let $D_2 = \otimes_{i=1}^{A_{d-1}} D_{2i}$ where for each $i \in A_{d-1}$, $D_{2i}$ samples $1$ with probability $t_{d-1}$ and $0$ otherwise. Let $\rho_d = \rhoin$ be the $q'$-block-random restriction with underlying distribution $D_1$, where each block is sampled from $\{*_{1/2},1_{1/2}\}^l\setminus\{1^l\}$ with probability $q'$ and from the corresponding block of $D_1$ otherwise. Clearly, $\rhoin \leftarrow \mathcal{R}_{\mathsf{init}}$. Let $D$ be the distribution induced by $D_2$ on the indices set to $*$ by $\mathsf{proj}_{\rho}$, so $D = \{0_{1-t_{d-1}},1_{t_{d-1}}\}^{(\widehat{\rhoin})^{-1}(*)}$. Let $C'$ be the $\AC^0_{d-1}$ circuit obtained from $C$ after applying $\mathsf{proj}_{\rhoin}$, by replacing each quantum gate with a width-$k^2$ DNF/CNF as per \cor{projdisagrdnf}. By a union bound,
    \begin{align*}
        \Pr_{\rhoin, z \sim D}[\mathsf{proj}_{\rhoin}C(z) \neq C'(z)] \leq S(\frac{k}{6}+2)e^{-k/6}
    \end{align*}
    We will think of the circuit $C'$ above as an $\AC^0_{d}$ circuit of bottom fan-in 1. Let $D_3 = \{0_{1-t_1},1_{t_1}\}^{w_0}$, and $\mathsf{OR}$ be the $\vee$ function on $w_0$ bits. For any restriction $\rho$ applied to $\mathsf{OR}$ and a string $y$ in $\{0,1\}^{w_0}$, we use $\mathsf{OR}_\rho(y)$ to mean $\mathsf{OR}_\rho$ evaluated at input $y_{|F}$ where $F$ is the set of indices set to $\ast$ by $\rho$. Now our proof is similar to that of \cite{HRST17}. We set $r = N^{\frac{1}{4(d-1)}}$.
    \begin{align*}
        \Pr_{\mathbf{X}}[\SIP'_{d}(\mathbf{X})\! \neq \!C(\mathbf{X})] &= \Pr_{\mathbf{\Psi}, y \sim D_3}[\mathbf{\Psi}(\SIP'_d(y) \neq \mathbf{\Psi}(C(y))] &\tag*{(Using \cor{uniformdist})} \\ \displaybreak
        &\geq \Pr_{\rho_d \cdots \rho_2, y \sim D_3}[\mathbf{\Psi}(\SIP'_d)(y) \neq \mathbf{\Psi}_{\rho_d}(C')(y)] - \!\!\!\Pr_{\rho_d, z \sim D}[\mathsf{proj}_{\rho_d}(C)(z) \neq C'(z)]\!\!\!\! \\
        &\geq \Pr_{\rho_d \cdots \rho_2, y \sim D_3}[\mathbf{\Psi}(\SIP'_d)(y) \neq \mathbf{\Psi}_{\rho_d}(C')(y)] - S(\frac{k}{6}+2)e^{-k/6} \\
        &= \Pr_{\rho_d \cdots \rho_2, y \sim D_3}[(\mathsf{OR}_{\widehat{\rho}_2})(y) \neq \mathbf{\Psi}_{\rho_d}(C')(y)] - S(\frac{k}{6}+2)e^{-k/6} \\
        &\geq \Pr_{\rho_d \cdots \rho_2, y \sim D_3}[(\mathsf{OR}_{\widehat{\rho}_2})(y) \neq \mathbf{\Psi}_{\rho_d}(C')(y)| \mathbf{\Psi}_{\rho_d}(C') \text{ is a width-$r$ CNF}] \\ 
        &-\Pr_{\rho_d \cdots \rho_2}[\mathbf{\Psi}_{\rho_d}(C') \text{ is not a width-$r$ CNF}] - S(\frac{k}{6}+2)e^{-k/6} \\
        &\geq \E_{\rho_d \cdots \rho_2}[\Pr_{y \sim D_3}[(\mathsf{OR}_{\widehat{\rho}_2})(y) \neq \mathbf{\Psi}_{\rho_d}(C')(y)]| \mathbf{\Psi}_{\rho_d}(C') \text{ is a width-$r$ CNF}] \\
        &-e^{-\Omega(N^{\frac{1}{6(d-1)}})} - S(\frac{k}{6}+2)e^{-k/6} &\tag*{(Using \cor{circuittocnf})} \\
        &= \E_{\mathbf{\Psi}}[\mathsf{bias}(\mathsf{OR}_{\widehat{\rho}_2}, \mathbf{y}) - rt_1] -e^{-\Omega(N^{\frac{1}{6(d-1)}})} - S(\frac{k}{6}+2)e^{-k/6} &\tag*{(Using \lem{corcnfor})} \\
        &= \E_{\mathbf{\Psi}}[\mathsf{bias}(\mathsf{OR}_{\widehat{\rho}_2}, \mathbf{y})] - rt_1 -e^{-\Omega(N^{\frac{1}{6(d-1)}})} - S(\frac{k}{6}+2)e^{-k/6} \\
        &\geq \frac{1}{2} - \Tilde{O}(w^{-1/12}) - rt_1 - e^{-\Omega(N^{\frac{1}{6(d-1)}})} - S(\frac{k}{6}+2)e^{-k/6} &\tag*{(Using \lem{unbiasedor})} \\
        &\geq \frac{1}{2} - \frac{1}{N^{\Omega(1/d)}} \qedhere
    \end{align*}
\end{proof}

An analogous proof follows for the case when $d \geq 3$ is odd, by replacing $\mathsf{OR}$ with $\mathsf{AND}$ (and $\wedge$ with $\vee$) and flipping the values 0 and 1 in $D_3$. Note also that the function $\SIP'_d$ has bottom fan-in $\polylog(N)$. Therefore, using \thm{randomlowerbound}, for every $d \geq 2$, we can conclude that relative to a random oracle $O$, $\Sigma_{d}^O \not \subset \QCPi_{d-1}^O$ when $d$ is odd, and $\Pi_{d}^O \not \subset \QCSigma_{d-1}^O$ when $d$ is even. In particular, for every $d \geq 2$ and relative to a random oracle $O$, we have that $\Sigma_{d}^O \not \subset \QCPi_{d-1}^O$. In addition, we can conclude that relative to a random oracle, $\QCPH$ is infinite.

\begin{theorem}\label{thm:randomoracleseparation}
    With probability 1, a random oracle O satisfies $\Sigma_{d}^O \not \subset \QCPi_{d-1}^O$ for odd $d \geq 3$.
\end{theorem}
\begin{proof}
    Proven in Appendix \ref{app:diagonalization}.
\end{proof}

\begin{corollary}\label{cor:randomoracleseparationeven}
    With probability 1, a random oracle O satisfies $\Pi_{d}^O \not \subset \QCSigma_{d-1}^O$ for even $d \geq 2$.
\end{corollary}
\begin{proof}
    Similar to \thm{randomoracleseparation}, using the analogue of \thm{randomlowerbound} for $\SIP'_d$ when $d$ is odd.
\end{proof}

\begin{corollary}\label{cor:qcphinfinite}
    With probability 1, a random oracle $O$ satisfies $\QCSigma_{d+1}^O \not \subset \QCPi_{d}^O$ for all $d \geq 1$, and therefore $\QCPH$ is infinite relative to $O$.
\end{corollary}
\begin{proof}
    Follows from \thm{randomoracleseparation}, \cor{randomoracleseparationeven} and \prop{qcphcollapse} (the proof of which relativizes).
\end{proof}

\paragraph{Remarks.} We now comment on how our switching lemma for quantum query algorithms (\thm{restbqpstronger}) compares to the switching lemma (Theorem 65) from \cite{AIK22}, and why the latter is not sufficient to establish our oracle separation result. In Theorem 65 of \cite{AIK22}, the authors show that after a uniformly random restriction, a $\QMA$-query algorithm is close in expectation to a DNF of small width, under the uniform distribution on the unrestricted bits. In fact, their proof also works if both the random restriction and the inputs for the unrestricted bits are sampled from a biased product distribution. However, this is not sufficient for our setting. From \defn{initdist}, we know that we choose the first random projection according to the distribution $\mathcal{R}_{\mathsf{init}}$. However, from \lem{uniformrst}, we know that the subsequent random projections complete to the $t_{d-1}$ biased product distribution. This happens because the projections are sampled from a biased underlying distribution, so that not all $\mathsf{AND}, \mathsf{OR}$ gates are set to a constant value $0$. Therefore, we need to choose the random projection for the quantum query algorithm from a different underlying distribution than the one for the inputs on unrestricted bits, using which we compare the closeness of the decision tree. Thus, in \thm{restbqpstronger}, we consider two inputs $x, y$, where $x$ is used to choose the restricted bits and $y$ is used to provide an input for the unrestricted bits, and these strings can be sampled from different distributions. A similar property is needed for proving the depth hierarchy theorem for the $\QCMA$ hierarchy in \sec{QCMAH}, where we repeatedly apply our projection switching lemma for quantum query algorithms with distributions of different biases on the restricted bits and the unrestricted input bits.

\section{Oracle Separation for \texorpdfstring{$\QCMA$}{QCMA} Hierarchy}\label{sec:QCMAH}
In this section, we show that there exists an oracle such that the $\QCMA$ hierarchy (called $\QCMAH$) is infinite, and no fixed level of $\QCMA$ hierarchy contains all of $\PH$. We first redefine the $\SIP$ function, which we will call $\SIP^{''}$. The function $\SIP^{''}_d$ will be an AND-OR tree of depth $2d+2$ (with AND gates at the bottom), and we will show that it can not be computed by a circuit corresponding to a $\QCMAH_d$ machine. Using the parameters below, we can define the fan-ins for the $\SIP^{''}_d$ function. The fan-in for each depth $j$ is denoted below by $w_j$. As before, the parameters $t_k$ correspond to the bias of the distribution for the random projection for depth $k$. The random projections are defined explicitly after \lem{niapprox}. The parameters $N_i$ roughly correspond to the number of gates (or variables) at depth $2i+2$, as shown in \lem{niapprox}.
\begin{align*}
    w &:= \lfloor m2^m/\log e \rfloor \\
    q &:= \sqrt{p} = 2^{-m/2} = \Theta \left(\sqrt{\frac{\log w}{w}}\right) \\
    N_1 &:= (w_0qw^3)^{\frac{7}{2}} \\
    N_i &:= (w_0qw^3)^{(\frac{7}{2})^i}(qw^2)^{\frac{7}{5}((\frac{7}{2})^{i-1}-1)} &\text{ for } 2\leq i \leq d-1 \\
    N_d &:= (w_0qw^3)^{(\frac{7}{2})^d}(qw^2)^{\frac{7}{5}((\frac{7}{2})^{d-1}-1)}(qw)^{\frac{7}{2}} \\
    q_i &:= \frac{1}{N_i^{\frac{5}{7}}} &\text{ for } 1 \leq i \leq d \\
    \lambda_i &:= \frac{q_i}{qw^{\frac{5}{4}}} &\text{ for } 1 \leq i \leq d \\
    \lambda &:= \frac{(\log w)^{3/2}}{w^{5/4}} \\
    t_{2d+1} &:= \frac{p-\lambda}{q} = q \pm q^{1.1} = \Theta \left(\sqrt{\frac{\log w}{w}}\right) \\ \displaybreak
    t_{k-1} &:= \frac{(1-t_k)^{qw}-\lambda}{q} = q \pm q^{1.1} = \Theta \left(\sqrt{\frac{\log w}{w}}\right) &\text{ for } 2\leq k \leq 2d+1 \\
    f_i &:= \frac{\log(\lambda_i + q_it_{2i+1})}{q\log(1-t_{2i+2})} &\text{ for } 1\leq i \leq d-1 \\
    f_d &:= \log(\lambda_d + q_dt_{2d+1}) = \polylog(N_d) \\ 
    \beta(k,2d+2) &:= \frac{1}{3} + \frac{2d-k+1}{24(d+1)} < \frac{5}{12} &\text{ for } 1\leq k \leq 2d+1 \\
    w_0 &:= \min_{i \in \mathbb{N}}\{(1-t_1)^{qi} \leq \frac{1}{2}\} \\
    w_1 &:= w \\
    w_{2i} &:= qwN_i^{\frac{5}{7}} &\text{ for } 1\leq i \leq d \\
    w_{2i+1} &:= f_i &\text{ for } 1\leq i \leq d \\
    N &= \Pi_{i=0}^{2d+1} w_i
\end{align*}
We now show that $f_i = \Tilde{\Theta}(w)$ for all $1 \leq i \leq d-1$, and therefore $N = \Tilde{\Theta}(N_d)$.
\begin{lemma}\label{lem:niapprox}
    For all $1 \leq i \leq d-1$, $f_i = \Tilde{\Theta}(w)$. Therefore, $N = \Tilde{\Theta}(N_d)$.
\end{lemma}
\begin{proof}
    We first establish the estimate on $f_i$:
    \begin{align*}
        \lambda_i + q_it_{2i+1} &= \Theta(\frac{q}{N_i^{5/7}}) \\
        q &= \Theta \left(\sqrt{\frac{\log w}{w}}\right) \\
        \Rightarrow N_i &= \Tilde{\Theta}(q^{\Theta(1)}) \\ 
        \Rightarrow \frac{q}{N_i^{5/7}} &= \Tilde{\Theta}(q^{-\Theta(1)}) \\
        \Rightarrow \lambda_i + q_it_{2i+1} &= -\Theta(\log(q))
    \end{align*}
    Further, since $t_{2i+2} \in (0, 1/2)$
    \begin{align*}
        -2t_{2i+2} \leq \log(1-t_{2i+2}) \leq -t_{2i+2} \\
        \Rightarrow f_i = \frac{1}{q} \Theta(\frac{\log(q)}{q}) = \Tilde{\Theta}(w)
    \end{align*}
    Let $N'_i = \Pi_{j=0}^{2i+1}w_j$ for $1 \leq i \leq d$. Note that $N = N'_d$. Using the estimate of $f_i$ above, we can conclude that $N'_i = \Tilde{\Theta}(N_i)$ for all $1 \leq i \leq d$.
\end{proof}
We will then apply $2d-1$ random projections, corresponding to the following random restrictions. For $1 \leq i \leq d$, we sample random restrictions $\rho_{i,1} \in \{0,1,\ast\}^{A_{2i+2}}$ which acts on quantum query algorithms and $\rho_{i,2} \in \{0,1,\ast\}^{A_{2i+1}}$ which acts on DNFs as follows:
\begin{itemize}
    \item The first restriction $\rho_{d,1}$ is sampled by sampling independently for each $a \in A_{2d+1}$:
    \begin{align*}
    \rho(a) \leftarrow \begin{cases}
            \{1\}^{f_d} & \text{with probability $\lambda_d$} \\
            \{*_{1/2},1_{1/2}\}^{f_d}\setminus\{1^{f_d}\} & \text{with probability $q_d$} \\
            \{0_{1/2},1_{1/2}\}^{f_d}\setminus\{1^{f_d}\} & \text{with probability $1-\lambda_d-q_d$}
        \end{cases}
\end{align*}
    \item For $1 \leq i \leq d-1$, given $\tau \in \{0,1,\ast\}^{A_{2i+2}}$, the restriction $\rho_{i,1}$ is sampled by sampling independently for each $a \in A_{2i+1}$:
    \begin{itemize}
        \item For all $(a, j) \in A_{2i+1} \times [w_{2i+1}]$, if $\tau_{a,j} \neq \ast$, then $\rho(a,j) = \tau_{a,j}$. In particular, if $\hat{\tau}(a) = 1$ then $\rho(a,j) = \tau_{a,j} = 1$ for all $j \in [w_{2i+1}]$.
        \item If $\hat{\tau}(a) = 0$ or if $|S_a| \neq qw_{2i+1} \pm w^{\beta(2i+2,2d+2)}$, then each bit of $\rho$ in $S_a$ is set independently to $0$ with probability $t_{2i+2}$ and $1$ with probability $1-t_{2i+2}$.
        \item Otherwise,
        \begin{align*}
    \rho(S_a) \leftarrow \begin{cases}
            \{1\}^{S_a} & \text{with probability $\lambda_i$} \\
            \{*_{t_{2i+2}},1_{1-t_{2i+2}}\}^{S_a}\setminus\{1^{S_a}\} & \text{with probability $q_{i,a}$} \\
            \{0_{t_{2i+2}},1_{1-t_{2i+2}}\}^{S_a}\setminus\{1^{S_a}\} & \text{with probability $1-\lambda_i-q_{i,a}$}
        \end{cases}
\end{align*}
    where
    \begin{align*}
        q_{i,a} := \frac{(1-t_{2i+2})^{|S_a|}-\lambda_i}{t_{2i+1}}
    \end{align*}
    \end{itemize}
    \item For $1 \leq i \leq d$, given $\tau \in \{0,1,\ast\}^{A_{2i+1}}$, the restriction $\rho_{i,2}$ is sampled by sampling independently for each $a \in A_{2i}$:
    \begin{itemize}
        \item For all $(a, j) \in A_{2i} \times [w_{2i}]$, if $\tau_{a,j} \neq \ast$, then $\rho(a,j) = \tau_{a,j}$. In particular, if $\hat{\tau}(a) = 0$ then $\rho(a,j) = \tau_{a,j} = 0$ for all $j \in [w_{2i}]$.
        \item If $\hat{\tau}(a) = 1$ or if $|S_a| \neq qw \pm w^{\beta(2i+1,2d+2)}$, then each bit of $\rho$ in $S_a$ is set independently to $0$ with probability $1-t_{2i+1}$ and $1$ with probability $t_{2i+1}$.
        \item Otherwise,
        \begin{align*}
    \rho(S_a) \leftarrow \begin{cases}
            \{1\}^{S_a} & \text{with probability $\lambda$} \\
            \{*_{1-t_{2i+1}},1_{t_{2i+1}}\}^{S_a}\setminus\{1^{S_a}\} & \text{with probability $q_{a}$} \\
            \{0_{1-t_{2i+1}},1_{t_{2i+1}}\}^{S_a}\setminus\{1^{S_a}\} & \text{with probability $1-\lambda-q_{a}$}
        \end{cases}
\end{align*}
    where
    \begin{align*}
        q_{a} := \frac{(1-t_{2i+1})^{|S_a|}-\lambda}{t_{2i}}
    \end{align*}
    \end{itemize}
    Note that these are the same restrictions as those sampled in \defn{subsequentproj}, we have restated them here to make the parameters explicit.
    \item The final random restriction acting on $A_2$ is also sampled according to \defn{subsequentproj}.
\end{itemize}
We now show that $q_{i,a}$ as defined for the restriction $\rho_{i,1}$ is well-defined and is indeed very close to $q_i$, for $1 \leq i \leq d-1$.

\begin{lemma}[Analogue of Lemma 10.5 of \cite{HRST17}]\label{lem:qiestimate}
    For $1 \leq i \leq d-1$, let $S_a$ be a set such that $|S_a| = qw_{2i+1} \pm w^{\beta(2i+2,2d+2)}$, and define
    \begin{align*}
        q_{i,a} := \frac{(1-t_{2i+2})^{|S_a|}-\lambda_i}{t_{2i+1}}
    \end{align*}
    Then $q_{i,a} = q_i(1 \pm 2t_{2i+2}w^{\beta(2i+2,2d+2)})$.
\end{lemma}
\begin{proof}
    We know from the statement that $|S_a| = qf_{i} \pm w^{\beta(2i+2,2d+2)}$. Therefore,
    \begin{align*}
        q_{i,a} &\leq \frac{(1-t_{2i+2})^{qf_i - w^{\beta(2i+2,2d+2)}} - \lambda_i}{t_{2i+1}} \\
        &= \frac{(1-t_{2i+2})^{qf_i} - \lambda_i(1-t_{2i+2})^{w^{\beta(2i+2,2d+2)}}}{t_{2i+1}(1-t_{2i+2})^{w^{\beta(2i+2,2d+2)}}} \\
        &\leq \frac{\lambda_i + q_it_{2i+1} - \lambda_i(1-t_{2i+2}w^{\beta(2i+2,2d+2)})}{t_{2i+1}(1-t_{2i+2}w^{\beta(2i+2,2d+2)})} \\
        &= \frac{q_i}{1-t_{2i+2}w^{\beta(2i+2,2d+2)}} + \frac{\lambda_it_{2i+2}w^{\beta(2i+2,2d+2)}}{t_{2i+1}(1-t_{2i+2}w^{\beta(2i+2,2d+2)})} \\
        &\leq \frac{q_i}{1-t_{2i+2}w^{\beta(2i+2,2d+2)}} + \frac{q_i}{qw^{5/4}}\frac{(1+3q^{0.1})w^{\beta(2i+2,2d+2)}}{(1-t_{2i+2}w^{\beta(2i+2,2d+2)})} \\
        &\leq q_i(1 + 2t_{2i+2}w^{\beta(2i+2,2d+2)})
    \end{align*}
    Here the second last inequality follows using the fact that $t_{2i+2}, t_{2i+1} = q \pm q^{1.1}$ The last inequality follows because $\frac{1+3q^{0.1}}{qw^{5/4}} = o_w(1)$. Further,
    \begin{align*}
        q_{i,a} &\geq \frac{(1-t_{2i+2})^{qf_i + w^{\beta(2i+2,2d+2)}} - \lambda_i}{t_{2i+1}} \\
        &\geq \frac{(1-t_{2i+2})^{qf_i}(1-t_{2i+2}w^{\beta(2i+2,2d+2)}) - \lambda_i}{t_{2i+1}} \\
        &\geq q_i(1-t_{2i+2}w^{\beta(2i+2,2d+2)}) - \frac{q_i}{qw^{5/4}t_{2i+1}} \\
        &\geq q_i(1-2t_{2i+2}w^{\beta(2i+2,2d+2)})
    \end{align*}
    Therefore the claim follows.
\end{proof}

A proof very similar to \lem{uniformrst} (from \cite{HRST17}) and \lem{uniform} can be used to conclude that the sequence of random projections sampled above completes to the uniform distribution. More strongly, the proof of \lem{uniformrst} tells us that for $1 \leq i \leq d$ projections upto $\rho_{i,2}$ complete to the $t_{2i+1}$ biased product distribution and projections upto $\rho_{i,1}$ complete to the $t_{2i+2}$ biased product distribution. This fact will be useful later to be able to apply the projection switching lemma for quantum algorithms repeatedly. 

In order to show that the $\SIP^{''}_d$ formula retains structure after applying each random projection, we modify the definition of typical restrictions (\defn{typical}) slightly, to account for the new parameters.

\begin{definition}[Typical random restriction]\label{def:typicalqcma}
    Let $1 \leq i \leq d$.
    \begin{itemize}
        \item The restriction $\rho_{i,1}$ is typical if
        \begin{enumerate}
        \item (Bottom fan-in after projection $\simeq qw$) For all $a \in A_{2i}$
        \begin{align*}
            |\hat{\tau}_a^{-1}(\ast)| = qw \pm w^{\beta(2i+1, 2d+2)}
        \end{align*}
        \item (Preserves rest of the structure) For all $a \in A_{2i-1}$
        \begin{align*}
            |(\hat{\hat{\tau}}_a)^{-1}(\ast)| \geq w_{2i-1} - w^{4/5}
        \end{align*}
    \end{enumerate}
        \item The restriction $\rho_{i,2}$ is typical if
        \begin{enumerate}
        \item (Bottom fan-in after projection $\simeq qw_{2i-1}$) For all $a \in A_{2i-1}$
        \begin{align*}
            |\hat{\tau}_a^{-1}(\ast)| = qf_{i-1} \pm w^{\beta(2i, 2d+2)}
        \end{align*}
        \item (Preserves rest of the structure) For all $a \in A_{2i-2}$
        \begin{align*}
            |(\hat{\hat{\tau}}_a)^{-1}(\ast)| \geq w_{2i-2}(1-w^{-1/5})
        \end{align*}
        \end{enumerate}
    \end{itemize}
\end{definition}
These conditions also imply that if a restriction $\rho$ acts on $A_k$, then all the gates in $A_{k-3}$ remain undetermined. This is because suppose $\rho$ is applied to a layer of $\wedge$ gates. Then by condition 1, the only values that $\vee$ gates in $A_{k-2}$ can get are $\ast$ or $1$ (so $\wedge$ gates in $A_{k-3}$ have inputs from $\ast$ and $1$). By condition 2, $\wedge$ gates in $A_{k-3}$ have at least one $\ast$ input, and since none of their inputs is $0$, they remain undetermined.

It follows from \lem{typical1} and \lem{typical2} that $\rho_{d,1}$ is typical with probability at least $1-e^{-\Omega(w^{1/6})}$. It follows from \lem{typicalrst} that $\rho_{i,2}$ is typical, for $1 \leq i \leq d$ (note that we have changed the fan-ins of the function and hence the definition of typical restriction, but the proof for \lem{typicalrst} works to establish this as well). We now show that $\rho_{i,1}$ is typical with  high probability, for $1 \leq i \leq d-1$. This proof is very similar to that of \lem{typicalrst} and can be adapted to show typicality for $\rho_{i,2}$ as well.

\begin{lemma}[Analogue of Lemma 10.6 of \cite{HRST17}]\label{lem:typicalcond1}
    For $1 \leq i \leq d-1$, given $\tau \in \{0,1,\ast\}^{A_{2i+2}}$ which is the lift of a typical restriction, the restriction $\rho_{i,1}$ is sampled. Then given $a \in A_{2i}$,
    \begin{align*}
        \Pr[|\hat{\rho_{i,1}}(a)^{-1}(\ast)| = qw \pm w^{\beta(2i+1, 2d+2)}] \geq 1-\exp(-\Tilde{\Omega}(w^{2\beta(2i+1,2d+2)-\frac{1}{2}})) \geq 1-e^{-\Omega(w^{1/6})}
    \end{align*}
\end{lemma}
\begin{proof}
    Since $\tau$ is the lift of a typical restriction,
    \begin{align*}
        |\hat{\tau}_a^{-1}(*)| \geq q(w-w^{\frac{4}{5}})N_i^{\frac{5}{7}}
    \end{align*}
    Further for all $j \in [w_{2i}]$ such that $\hat{\tau}_{a,j} = \ast$,
    \begin{align*}
        |S_{a,j}| = |(\tau_{a,j})^{-1}(*)| = qf_i \pm w^{\beta(2i+2, 2d+2)}
    \end{align*}
    Therefore using \lem{qiestimate}, we know that
    \begin{align*}
        q_{i,(a,j)} = q_i(1 \pm 2t_{2i+2}w^{\beta(2i+2,2d+2)})
    \end{align*}
    Therefore,
    \begin{align*}
        \E[|\hat{\rho_{i,1}}(a)^{-1}(\ast)|] &\leq qwN_i^{\frac{5}{7}}.q_i(1 \pm 2t_{2i+2}w^{\beta(2i+2,2d+2)}) \\
        &= qw(1 \pm 2t_{2i+2}w^{\beta(2i+2,2d+2)}) \\
        &\leq qw + \Tilde{O}(w^{\beta(2i+2, 2d+2)})
    \end{align*}
    where the last inequality holds because $qwt_{2i+2} = \Theta(\log(w))$. Further,
    \begin{align*}
        \E[|\hat{\rho_{i,1}}(a)^{-1}(\ast)|] &\geq q(w-w^{\frac{4}{5}})N_i^{\frac{5}{7}}.q_i(1 \pm 2t_{2i+2}w^{\beta(2i+2,2d+2)}) \\
        &\geq qw - \Tilde{O}(w^{\beta(2i+2, 2d+2)})
    \end{align*}
    where the last inequality follows because $w^{\frac{4}{5}}q = o(w^{\beta(2i+2, 2d+2)}$. Therefore on applying a Chernoff bound (\prop{chernoff}),
    \begin{align*}
        \Pr[|\hat{\rho_{i,1}}(a)^{-1}(\ast)| = qw \pm w^{\beta(2i+1, 2d+2)}] &\geq 1-\exp(-\Omega(w^{2\beta(2i+1, 2d+2)/qw})) \\
        &\geq 1-\exp(-\Tilde{\Omega}(w^{2\beta(2i+1,2d+2)-\frac{1}{2}}))
    \end{align*}
    Therefore, the claim follows.
\end{proof}

\begin{lemma}[Analogue of Lemma 10.8 of \cite{HRST17}]\label{lem:typicalcond2}
    For $1 \leq i \leq d-1$, given $\tau \in \{0,1,\ast\}^{A_{2i+2}}$ which is the lift of a typical restriction, the restriction $\rho_{i,1}$ is sampled. Then given $a \in A_{2i-1}$,
    \begin{align*}
        \Pr[|(\hat{\hat{\rho}}_a)^{-1}(\ast)| \geq w_{2i-1} - w^{4/5}] \geq 1-e^{-\Tilde{\Omega}(w^{4/5})}
    \end{align*}
\end{lemma}
\begin{proof}
    Since $\tau$ is the lift of a typical restriction,
    \begin{align*}
        |\hat{\tau}_a^{-1}(*)| \geq q(w-w^{\frac{4}{5}})N_i^{\frac{5}{7}}
    \end{align*}
    Further for all $j \in [w_{2i}]$ such that $\hat{\tau}_{a,j} = \ast$,
    \begin{align*}
        |S_{a,j}| = |(\tau_{a,j})^{-1}(*)| = qf_i \pm w^{\beta(2i+2, 2d+2)}
    \end{align*}
    Therefore using \lem{qiestimate}, we know that
    \begin{align*}
        q_{i,(a,j)} = q_i(1 \pm 2t_{2i+2}w^{\beta(2i+2,2d+2)}) \geq \frac{q_i}{2}
    \end{align*}
    Now look at an OR gate at address $(a,j)$ for $j \in [w_{2i-1}]$. This gate is not determined to value 1 with probability $(1-\lambda_i)^{w_{2i}} \geq (1-\lambda_iw_{2i}) \geq 1-w^{-1/4}$. This gate is determined to 0 with probability at most $(1-\lambda_i - q_{i,(a,j)})^{q(w-w^{\frac{4}{5}})N_i^{\frac{5}{7}}} \leq (1-\frac{q_i}{2})^{q(w-w^{\frac{4}{5}})N_i^{\frac{5}{7}}} \leq e^{-qw/4} \leq w^{-1/4}$. Therefore,
    \begin{align*}
        \E[|(\hat{\hat{\rho}}_a)^{-1}({0,1})|] \leq w_{2i-1}O(w^{-1/4}) \leq \Tilde{O}(w^{3/4})
    \end{align*}
    Therefore on applying a Chernoff bound (\prop{chernoff})
    \begin{align*}
        \Pr[|(\hat{\hat{\rho}}_a)^{-1}(\ast)| \geq w_{2i-1} - w^{4/5}] \geq 1-e^{-\Tilde{\Omega}(w^{4/5})}
    \end{align*}
    Therefore the claim follows.
\end{proof}

Now we show an analogue of \cite{furstParityCircuitsPolynomialtime1984} for the $\QCMA$ hierarchy, to use a circuit lower bound to get an oracle separation result. Note that \cite{AIK22} show a similar result for $\QMA$-hierarchy.
\begin{proposition}[Analogue of \cite{furstParityCircuitsPolynomialtime1984}]\label{prop:qcmahcirc}
    Let $L \subseteq \{0\}^{*}$ be some language decided by a $\QCMAH_i$ verifier $V = \langle V_1, \ldots, V_i \rangle$ with oracle access to $O$, and where each $V_j$ has size at most $p(n)$ for inputs of length $n$. Let $q_i(n) = p(p(\ldots p(n)\ldots )$ where $p$ is composed $i$ times. Then for every $n \in \mathbb{N}$, there is a circuit $C$ of size at most $2^{\poly (n)}$ and depth $2i$, where the top layer is an OR gate and the layers alternate between OR gates of width $2^{q_i(n)}$, and quantum circuits of size $q_i(n)$ with query complexity at most $q_i(n)$ such that $\forall~x \in \{0,1\}^n$
    \begin{align*}
        V^O(x) = C(O_{[\leq q_i(n)]})
    \end{align*}
    where $O_{[\leq q_i(n)]}$ denotes the concatenation of bits of $O$ on all strings of length at most $q_i(n)$.
\end{proposition}
\begin{proof}
    This proof is analogous to that of \cite{AIK22} for the $\QMA$-hierarchy. We proceed by induction. For $i = 1$, $V$ is just a $\QCMA$ verifier and hence can be represented by a circuit $C$ of size at most $2^{\poly (n)}$ and depth $2$, where the top layer is an OR gate of width $2^{p(n)}$ and the bottom (second layer) is a quantum circuit of size $p(n)$ with query complexity $p(n)$, the output of which depends only on the bits of $O$ on strings of length at most $p(n)$. This is similar to \prop{qcphcirc}. \\
    For $i \geq 2$, the circuit $C$ consists of an OR gate of width $p(n)$ followed by a quantum circuit of size $p(n)$ with query complexity $p(n)$, which by the induction hypothesis, queries functions computable by circuits of depth $2i-2$ with alternating layers of OR gates of width $2^{q_{i-1}(p(n))} = 2^{q_i(n)}$ and quantum query gates of size at most $q_{i-1}(p(n)) = q_i(n)$ and query complexity at most $q_{i-1}(p(n)) = q_i(n)$. The size of this circuit will be $2^{\poly(n)}$ for a reasonably large polynomial, since the circuit has constant depth.
\end{proof}
The proof for the depth-hierarchy theorem now follows very similarly to the proof for $\QCPH$ (or $\QAC^0$ circuits) in \thm{randomlowerbound}. We noted earlier that the sequence of random projections completes to an appropriately biased product distribution at every level of the AND-OR tree, hence we can apply the projection switching lemma for quantum algorithms which works for arbitrary inner and outer product distributions. Since we know that $q_{i,a}$ is very close to $q_i$ and that $N'_i = \Tilde{\Theta}(N_i)$ (from \lem{niapprox}), we can use the projection switching lemma for quantum algorithms querying $N'_i$ bits on setting the probability of leaving a block unrestricted to $q_{i,a}$. This switching lemma and the projection switching lemma for DNFs \thm{randomprojcnf} by \cite{HRST17} are applied in an alternating manner to get a depth-hierarchy theorem (the projection switching lemma for quantum algorithms is applied for $\rho_{i,1}$ and the projection switching lemma for DNFs is applied for $\rho_{i,2}$). Then oracle separations analogous to the case for $\QCPH$ follow for $\QCMAH$ using a standard diagonalization argument. Note however that now we will get an oracle separation between $\Sigma_{2d+1}$ and $\QCMAH_d$ relative to a random oracle $O$ (whereas for $\QCPH$ we got an oracle separation between $\Sigma_{d+1}$ and $\QCPi_d$).

\if\anon1
\else
\phantomsection\addcontentsline{toc}{section}{Acknowledgements}
\section*{Acknowledgements}
This research is supported in part by the Natural Sciences and
Engineering Research Council of Canada (NSERC), DGECR-2019-00027 and
RGPIN-2019-04804.\footnote{Cette recherche a été financée par le
Conseil de recherches en sciences naturelles et en génie du Canada
(CRSNG), DGECR-2019-00027 et RGPIN-2019-04804.}
\fi


\appendix

\section{The diagonalization argument}\label{app:diagonalization}
We now show how to use \thm{randomlowerbound} to show that $\Sigma_{d+1}^O \not \subset \QCSigma_d^O$ relative to a random oracle. The proof follows the standard argument, which has appeared before in \cite{furstParityCircuitsPolynomialtime1984}, \cite{DBLP:conf/stoc/Hastad86}, \cite{rstoverview}, \cite{razOracleSeparationBQP2019}, \cite{AIK22}.
\begin{theorem}
    With probability 1, a random oracle O satisfies $\Sigma_{d}^O \not \subset \QCPi_{d-1}^O$ for odd $d \geq 3$.
\end{theorem}
\begin{proof}
    Relative to any oracle $O$, define the unary language
    \begin{align*}
        L^O = \{1^n | \SIP'_{d+1}(O_{[n]}) = 1\}
    \end{align*}
    where $O_{[n]}$ denotes the concatenation of outputs of $O$ on all inputs of length $n$. Since the bottom fan-in of $\SIP'_{d+1}$ on $N = 2^n$ variables is $\polylog(N) = \poly(n)$, we can conclude that $L^O \in \Sigma_{d}^O$. We now want to show that with probability 1, for a random oracle $O$, $L^O \notin \QCPi_{d-1}^O$. Let $M^O$ be a $\QCPi_{d-1}^O$ oracle machine. We will show that
    \begin{align*}
        \Pr_O[M^O \text{ decides } L^O] = 0
    \end{align*}
    Let $\{n_i\}_{i=1}^\infty$ be an increasing sequence of natural numbers such that there exists $j \in \bN$ such that $M^O(1^{n_i})$ can not query $O$ on inputs of length $\geq n_{i+1}$ for $i \geq j$. Note that such a sequence can always be constructed using the upper bound on running time of $M^O$. Then let
    \begin{align*}
       p(M^O, i) := \Pr_O[M^O \text{ decides } 1^{n_i}|M^O \text{ decides } 1^{n_k}~\forall 1 \leq k \leq i-1]
    \end{align*}
    Therefore,
    \begin{align*}
        \Pr_O[M^O \text{ decides } L^O] \leq \prod_{i=1}^\infty p(M^O, i)
    \end{align*}
    So we need to show that $p(M,i) \leq c$ for some constant $c < 1$ for all sufficiently large $i \in \bN$. We know from \thm{randomlowerbound} that for sufficiently large $l \in \bN$
    \begin{align*}
        \Pr_O[M^O \text{ decides } 1^{n_i}] \leq 0.6~~~~~~~~\forall~i\geq l
    \end{align*}
    By choice of $\{n_i\}_{i=1}^\infty$,
    \begin{align*}
        p(M^O, i) = \Pr_O[M^O \text{ decides } 1^{n_i}] \leq 0.6
    \end{align*}
    for $i \geq \max\{j,l\}$, since $O_{[n]}$ is sampled independently for every $n \in \bN$ and $M^O$ can not make queries of length $n_i$ or greater, on inputs $1^{n_k}$ for $k \leq i-1$. Thus we can conclude that
    \begin{align*}
        \Pr_O[M^O \text{ decides } L^O] \leq \prod_{i=1}^\infty p(M^O, i) = 0
    \end{align*}
    Since there are countably many such $\QCPi_{d-1}^O$ oracle machines, we are done.
\end{proof}

\renewcommand{\UrlFont}{\ttfamily\small}
\let\oldpath\path
\renewcommand{\path}[1]{\small\oldpath{#1}}
\emergencystretch=2em 

\if\bibnonempty1{ 
\phantomsection\addcontentsline{toc}{section}{References} 
\printbibliography 
}\else{
\nocite{AA05} 
}
\fi
\end{document}